\documentclass[aps,prx,twocolumn,superscriptaddress,longbibliography,10pt]{revtex4-2}

\usepackage{tikz-cd}
\usepackage{mathtools}
\usepackage{color}

\usepackage{silence}
\usepackage{bbm}
\usepackage{amssymb}
\usepackage{microtype}
\usepackage{xcolor}
\usepackage{enumitem}
\usepackage{amsthm}
\usepackage{physics}
\usepackage{proof}
\usepackage{thmtools,thm-restate}
\usepackage{graphicx}
\usepackage{algorithmicx}
\usepackage{algpseudocode}
\usepackage{qcircuit}
\usepackage{braket}
\usepackage{multirow}
\usepackage{mathrsfs}
\usepackage{soul}
\usepackage{amsmath}
\usepackage{amsfonts}
\usepackage{url}
\usepackage{ragged2e}
\DeclareMathAlphabet{\mathpzc}{OT1}{pzc}{m}{it}

\usepackage[colorlinks,citecolor=blue]{hyperref}
\newtheorem{theorem}{Theorem}[section]
\newtheorem{lemma}[theorem]{Lemma}
\newtheorem{definition}[theorem]{Definition}
\newtheorem{proposition}[theorem]{Proposition}
\newtheorem{corollary}[theorem]{Corollary}

\newtheorem{remark}[theorem]{Remark}

\newcommand{\eq}[1]{\hyperref[eq:#1]{(\ref*{eq:#1})}}
\renewcommand{\sec}[1]{\hyperref[sec:#1]{Section~\ref*{sec:#1}}}
\newcommand{\thm}[1]{\hyperref[thm:#1]{Theorem~\ref*{thm:#1}}}
\newcommand{\lem}[1]{\hyperref[lem:#1]{Lemma~\ref*{lem:#1}}}
\newcommand{\cor}[1]{\hyperref[cor:#1]{Corollary~\ref*{cor:#1}}}
\newcommand{\itm}[1]{\hyperref[itm:#1]{\ref*{itm:#1}}}
\newcommand{\app}[1]{\hyperref[app:#1]{Appendix~\ref*{app:#1}}}
\newcommand{\dfn}[1]{\hyperref[dfn:#1]{Definition~\ref*{dfn:#1}}}
\newcommand{\fig}[1]{\hyperref[fig:#1]{Figure~\ref*{fig:#1}}}
\newcommand{\clm}[1]{\hyperref[clm:#1]{Claim~\ref*{clm:#1}}}
\newcommand{\alg}[1]{\hyperref[alg:#1]{Algorithm~\ref*{alg:#1}}}
\newcommand{\stp}[1]{\hyperref[stp:#1]{Step~\ref*{stp:#1}}}
\newcommand{\asm}[1]{\hyperref[asm:#1]{Assumption~\ref*{asm:#1}}}
\newcommand{\prot}[1]{\hyperref[prot:#1]{Protocol~\ref*{prot:#1}}}
\newcommand{\prob}[1]{\hyperref[prob:#1]{Problem~\ref*{prob:#1}}}
\newcommand{\rmk}[1]{\hyperref[rmk:#1]{Remark~\ref*{rmk:#1}}}
\newcommand{\cons}[1]{\hyperref[cons:#1]{Construction~\ref*{cons:#1}}}
\newcommand{\conj}[1]{\hyperref[conj:#1]{Conjecture~\ref*{conj:#1}}}
\newcommand{\tbl}[1]{\hyperref[tbl:#1]{Table~\ref*{tbl:#1}}}
\usepackage{tabularx}

\let\originalleft\left
\let\originalright\right
\renewcommand{\left}{\mathopen{}\mathclose\bgroup\originalleft}
\renewcommand{\right}{\aftergroup\egroup\originalright}

\begin{document}

\title{Quantum Magic in early FTQC: From Diagonal Clifford Hierarchy No-Go Theorems to Architecture Design Blueprints}

\author{Hsueh-Hao Lu}
\email{hsuehhao-lu@g.ecc.u-tokyo.ac.jp
}
\affiliation{Department of Applied Physics, Graduate School of Engineering, The University of Tokyo, Bunkyo-ku, Tokyo 113-8656, Japan}
\affiliation{RIKEN Center for Quantum Computing (RQC), Wako, Saitama 351-0198, Japan}
\author{Yasunari Suzuki}
\affiliation{Department of Applied Physics, Graduate School of Engineering, The University of Tokyo, Bunkyo-ku, Tokyo 113-8656, Japan}
\affiliation{RIKEN Center for Quantum Computing (RQC), Wako, Saitama 351-0198, Japan}
\author{Yasunobu Nakamura}
\affiliation{Department of Applied Physics, Graduate School of Engineering, The University of Tokyo, Bunkyo-ku, Tokyo 113-8656, Japan}
\affiliation{RIKEN Center for Quantum Computing (RQC), Wako, Saitama 351-0198, Japan}
\author{En-Jui Kuo}
\email{kuoenjui@nycu.edu.tw}
\affiliation{Department of Electrophysics, National Yang Ming Chiao Tung University, Hsinchu, Taiwan, R.O.C.}
\begin{abstract}
We address the circuit-design problem of maximizing quantum magic in early fault-tolerant quantum computing (early FTQC), where logical dynamics natively take the form of alternating Clifford layers and diagonal non-Clifford layers. To render this optimization analytically tractable, we first prove a uniqueness theorem: for operational magic functionals built from Pauli expectation values, the axioms of faithfulness and tensor-product additivity force a Rényi-type dependence on the Pauli-spectrum moments. Leveraging the closed phase-polynomial description of the diagonal Clifford hierarchy, we derive exact Pauli-spectrum expressions and tight bounds for a shallow-layer model (ansatz~1). These bounds expose a zero-magic mechanism and prove that maximal magic strictly requires graph-state preconditioning. Consequently, we establish our first no-go theorem: hierarchy level alone cannot universally order operational magic. Extending our framework to the $N$-layer model (ansatz~2) motivated by the Space-Time Efficient Analog Rotation (STAR) architecture, we obtain an exact iterative update rule for the Pauli spectrum. This yields a second no-go theorem: no state-independent sequence of operations can guarantee monotonic magic improvement. Together, these theorems demonstrate that algebraic gate structures are fundamentally insufficient to dictate resource generation. To overcome this, we reframe early FTQC gate selection as a state-aware, differentiable optimization over continuous analog parameters. Finally, we identify a severe kinematic expressibility bottleneck in architectures restricted to single-qubit $Z$-rotations and rigorously show that introducing nonlinear diagonal phases, such as multi-qubit $Z$-rotation, shatters this bottleneck. This provides a fundamental principle for demonstrating early FTQC, establishing scalable magic generation as a foundational benchmark for evaluating early FTQC architectures.
\end{abstract}

\maketitle

\section{Introduction}
Quantum computing has the potential to demonstrate quantum advantage~\cite{PhysRevLett.127.180501,Arute:2019zxq,Zhong_2020}. Circuits built from stabilizer operations, such as Clifford unitaries, Pauli measurements, feedforward, and stabilizer ancillas, admit efficient classical simulation by the Gottesman–Knill theorem~\cite{gottesman_knill}, whereas non-stabilizer resources are required to transcend this simulability and enable universal quantum computing~\cite{bravyi_kitaev,Mari_2012}. At the architectural level, this distinction is equally consequential: in fault-tolerant quantum computing (FTQC), Clifford operations are implemented with relatively low overhead via lattice surgery~\cite{Horsman_2012,Litinski_2019}, whereas non-Clifford operations remain the dominant resource bottleneck due to magic-state distillation costs~\cite{Bravyi_2012}. Crucially, a gate's capacity to generate non-stabilizer resources is not an intrinsic property of the operation alone—it depends sensitively on the input state and circuit structure—rendering gate-level resource engineering a highly nontrivial design challenge~\cite{heyfron2018efficientquantumcompilerreduces}.

The resource theory of quantum magic provides a quantitative framework for non-stabilizerness by introducing magic monotones, which vanish on stabilizer states~\cite{veitch2012negative} and bound the cost of operational tasks such as classical simulation and state synthesis~\cite{Howard_2017,Beverland_2020,Bravyi_2019}. In particular, a wide class of simulation algorithms directly links the accumulation of magic to an exponential increase in the classical computational cost required to sample from or strongly simulate the underlying quantum processes. Therefore, understanding how logical circuits generate and redistribute magic is essential not only for delineating the boundary between classically simulable and hard quantum dynamics but also for guiding the efficient generation of non-stabilizer resources under realistic architectural constraints~\cite{beverland2022assessingrequirementsscalepractical}.

Algebraically, non-Clifford operations are systematically organized by the Clifford hierarchy~\cite{Gottesman_1999}, which classifies unitaries based on how they transform Pauli operators under conjugation. The third level contains widely used non-Clifford primitives such as the $T$ and $CCZ$ (Toffoli-equivalent) gates, which dominate standard fault-tolerant architectures since they admit well-developed injection and distillation protocols~\cite{Bravyi_2012}. By contrast, higher-level hierarchy operations—especially diagonal rotations that appear naturally in quantum algorithms—are usually realized only indirectly by compiling them into deep sequences of third-level primitives~\cite{ross2016optimalancillafreecliffordtapproximation}, for example via Clifford+$T$ synthesis of $Z$-rotations. This practice effectively obscures the hierarchy at execution time and leaves open a basic resource question: how does the hierarchy level of a gate relate to the amount of magic it can generate in a circuit?

The advent of early FTQC brings this question to the forefront. In this regime, Clifford operations are error-corrected while certain non-Clifford rotations can be executed as native logical instructions without the prohibitive overhead of full distillation, as exemplified by the space–time efficient analog rotation (STAR) architecture~\cite{PRXQuantum.5.010337}. STAR-like architectures motivate a restricted but hardware-relevant model in which the native non-Clifford resources are single-qubit analog Z-rotations. As a result, logical circuits naturally decompose into alternating layers of Clifford and diagonal gates $\{C_i,D_i\}$, and higher-level diagonal operations become operationally relevant rather than merely algebraic. In such cases, discrete metrics like $T$-count are ill-suited for these architectures; therefore, explicit magic quantification emerges as the rigorous benchmark for an architecture's capacity to induce classical simulation hardness. This setting leads to a concrete design problem: \textit{given an architectural template and constraints, how should one choose a sequence of logical operations $\{C_i,D_i\}$ to maximize the magic of the output state?} Derived from gate synthesis costs~\cite{ross2016optimalancillafreecliffordtapproximation,gosset2025multiqubittoffoliexponentiallyfewer, Beverland_2020}, the widespread monotonicity hypothesis assumes that higher-hierarchy operations inherently generate more magic, implying a universal, state-independent ordering strategy.

\begin{figure}[t]
        \includegraphics[width=\linewidth]{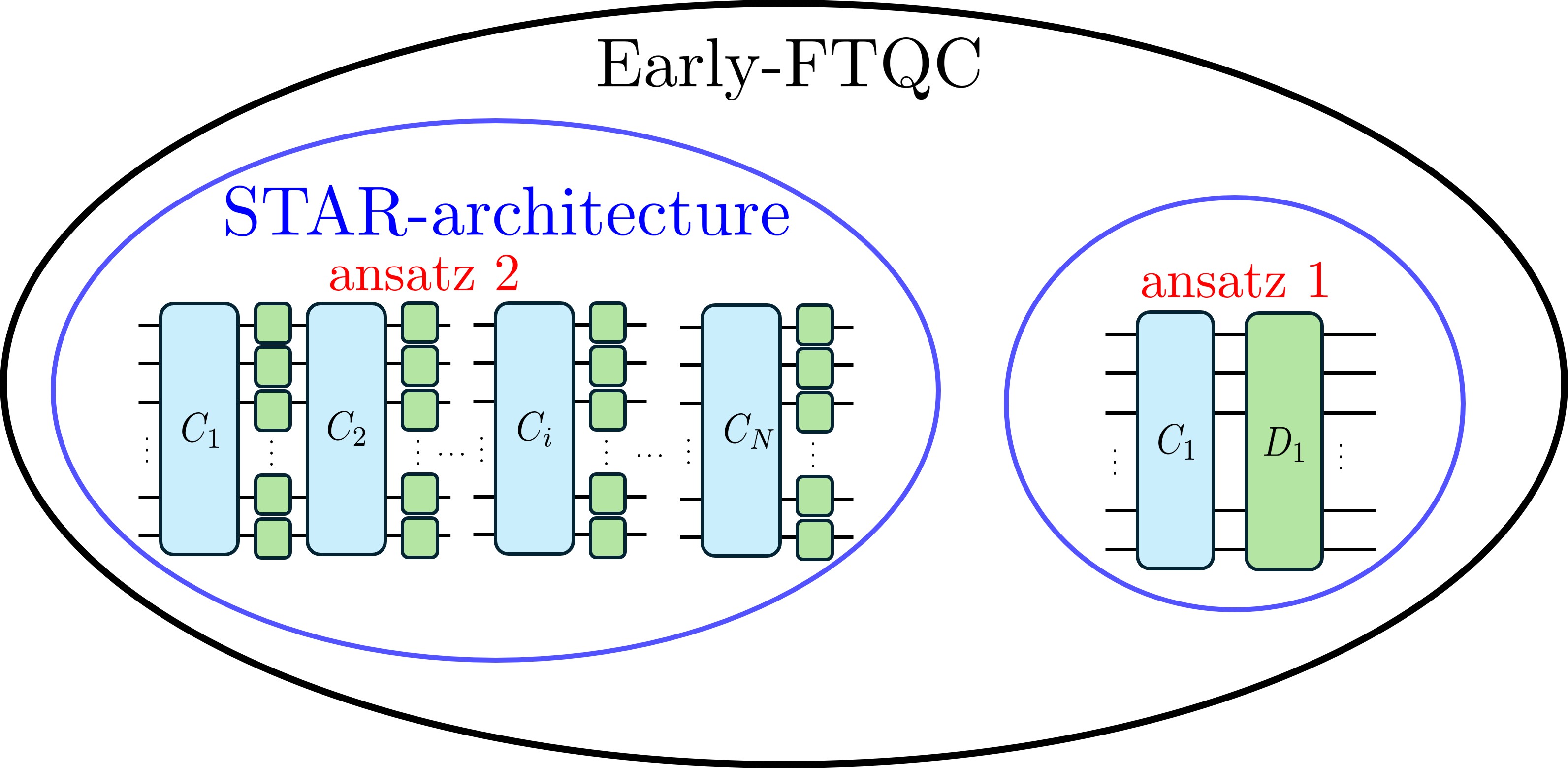}
  \caption{Structural classification of the quantum circuit architectures analyzed in this work. The overarching domain represents the capabilities of early fault-tolerant quantum computing (FTQC). The shallow-layer model (ansatz~1) utilizes a generalized multi-qubit diagonal gate ($D_1$) following a Clifford operation ($C_1$). This model captures the broad algebraic capacity of early FTQC but operates outside the restrictive Space-Time efficient Analog Rotation (STAR) paradigm (blue boundary). In contrast, ansatz~2 explicitly belongs to the STAR subset, utilizing an $N$-layer alternating sequence of Clifford operations ($C_i$) and localized single-qubit diagonal $Z$-rotation. Although kinematically distinct, both frameworks are analytically united by the shared algebraic properties of the diagonal Clifford hierarchy.}
  \label{Fig:1}
\end{figure}

In this work, we establish a unified resource-theoretic framework for analyzing and optimizing quantum magic in early FTQC architectures and provide concrete circuit design guidelines to maximize magic generation. Section~\ref{sec:bg_CH_DCH} introduces the Clifford hierarchy and its diagonal subgroup, while Sec.~\ref{sec:bg_nisq_ftqc} outlines the structural constraints of early FTQC and the hardware architectures considered in this work.

To establish the analytical machinery for our subsequent circuit evaluations, Sec.~\ref{sec:magic_category} classifies pure state magic measures into two main families. By focusing on measures constructed from Pauli expectation values, we prove that enforcing tensor-product additivity strictly reduces these quantities to logarithmic functionals of the Pauli spectrum, characterized by $\log F_\alpha$. This reduction provides a powerful, universal mathematical toolset: any structural constraint derived for $F_\alpha$ directly governs this entire class of experimentally relevant magic monotones~\cite{xiao2026exponentially,bittel2025operationalinterpretationstabilizerentropy}.

In Sec.~\ref{sec:shallow_layer}, we analyze a shallow layer model consisting of a Clifford layer followed by a diagonal $k$th-level Clifford Hierarchy gate in Fig.~\ref{Fig:1} (ansatz~1). By deriving exact expressions for the Pauli spectrum, we establish tight lower and upper bounds on magic generation. These bounds reveal a fundamental obstruction: the achievable ranges of magic overlap across hierarchy levels, leading to a no-go theorem—quantum magic cannot induce a strict ordering of the Clifford hierarchy.

In Sec.~\ref{sec:deep_layer}, we extend this analysis to an $N$-layer model motivated by STAR and early FTQC architectures in Fig.~\ref{Fig:1} (ansatz~2). We show that magic generation is intrinsically state-dependent and prove a second no-go theorem: no state-independent gate-selection strategy can guarantee monotonic enhancement of magic across circuit layers. Instead, optimal circuit design must be adapted to the evolving Pauli spectrum, which leads to a continuous optimization framework for selecting $\{C_i, D_i\}$.

Finally, we identify a fundamental architectural limitation of single-qubit $Z$ rotation~(SQR) in early FTQC architecture: restricting operations to single-qubit diagonal rotations severely limits magic generation. Consequently, an architecture's capacity to generate operational magic—and thereby breach the boundary of classical simulability—must serve as a foundational benchmark for evaluating early FTQC hardware design. This motivates a concrete design principle—introducing nonlinear phase interactions, such as multi-qubit $Z$ rotations, to enhance non-stabilizer resource generation.

\section{Clifford Hierarchy and Diagonal Clifford Hierarchy}
\label{sec:bg_CH_DCH}

The Clifford hierarchy provides a systematic classification of quantum operations based on their action on the Pauli group. Let $\mathcal{P}_n$ denote the $n$-qubit Pauli group. The hierarchy $\mathcal{C}_n^{(k)}$ is defined recursively as~\cite{Gottesman_1999}
\begin{align}
    &\mathcal{C}_n^{(1)} := \mathcal{P}_n, \notag \\
    &\mathcal{C}_n^{(2)} := \{ U \in U(2^n) \mid U \mathcal{P}_n U^\dagger \subseteq \mathcal{P}_n \}, \notag \\
    &\mathcal{C}_n^{(k+1)} := \{ U \in U(2^n) \mid U \mathcal{P}_n U^\dagger \subseteq \mathcal{C}_n^{(k)} \}.
    \label{eq:Clifford_Hierarchy}
\end{align}
Clifford operations ($k=2$) map Pauli operators to Pauli operators, while higher levels generate increasingly nontrivial transformations that cannot be efficiently simulated classically.

For $k \geq 3$, the structure of $\mathcal{C}_n^{(k)}$ is defined only implicitly through this recursion. A notable exception is the diagonal subgroup $\mathcal{D}_n^{(k)} \subset \mathcal{C}_n^{(k)}$, which admits a closed-form representation in terms of phase polynomials~\cite{Cui_2017}. Any $U \in \mathcal{D}_n^{(k)}$ can be written as
\begin{equation}
    U = \sum_{\mathbf{b}\in \mathbb{Z}_2^n} 
    \exp\left(2\pi i \sum_{m=1}^{m'} \frac{f_m(\mathbf{b})}{2^m}\right) 
    |\mathbf{b}\rangle\langle\mathbf{b}|,
    \label{eq:diagonal_hierarchy}
\end{equation}
where $m'$ is an integer representing the maximum phase resolution and each $f_m(\mathbf{b})$ is a Boolean polynomial of the form:
\begin{equation}
    f_m(\mathbf{b}) = \sum_{\mathbf{a} \in \{0,1\}^n} 
    c_{m,\mathbf{a}} \prod_{i=1}^n b_i^{a_i}, 
    \quad c_{m,\mathbf{a}} \in \mathbb{Z}_{2^m}.
    \label{eq:phase_polynomial}
\end{equation}
Each monomial encodes a multi-qubit phase interaction. The hierarchy level $k$ is determined by evaluating the maximum combined contribution of the interaction degree, given by the Hamming weight $\mathrm{wt}(\mathbf{a})$, and the associated phase-resolution index $m$:
\begin{equation}
    k = \max \{ (m-1) + \mathrm{wt}(\mathbf{a}) \},
    \label{eq:hierarchy}
\end{equation}
where the maximum is taken over all non-vanishing terms in the phase polynomial $c_{m,\{\mathbf{a}\}}$. would be clearer.  Physically, this quantity captures the interplay between the number of qubits involved in a controlled operation and its phase precision. 

To illustrate this, consider the composite operation $T^{\otimes 3} \cdot CCS$ where three transversal $T$ gates are applied alongside a controlled-controlled-$S$ gate. The maximum phase resolution for this unitary is $m'=3$, and there are three Boolean polynomials. Each of them can be derived as follows.
\begin{align}
    f_2(b) & = b_1b_2b_3\notag \\
    f_3(b) & = b_1 + b_2 + b_3
\end{align}
Therefore, $c_{3, (1,0,0)}$, $c_{3, (0,1,0)}$, $c_{3, (0,0,1)}$, and $c_{2,(1,1,1)}$ are non-vanishing terms in the phase polynomial. A term that maximizes Eq.~\eqref{eq:hierarchy} is $c_{2, (1,1,1)}$, which results in k = (2-1)+3=4. Thus, this unitary is in the fourth Clifford hierarchy.

In the STAR-architecture, arbitrary-angle Pauli-$Z$ rotations are implemented directly as native operations, effectively realizing continuous families of diagonal gates across multiple hierarchy levels~\cite{PRXQuantum.5.010337}. As a result, the diagonal Clifford hierarchy is not merely a mathematical subset, but the operational gate set of the architecture. It provides both an analytically tractable and physically relevant framework for studying non-stabilizer resource generation. 

A highly constrained, strictly local subset of this diagonal hierarchy is the continuous single-qubit $Z$-rotation ~(SQR). Restricting non-Clifford layers exclusively to SQRs ensures that the diagonal gates take a strict tensor-product form:
\begin{align}
    U = \bigotimes_{j=1}^n \mathrm{diag}\left(1, e^{2\pi i w_j}\right) 
      = \sum_{\mathbf{b}\in\mathbb{Z}_2^n} e^{2\pi i\mathbf{w}\cdot\mathbf{b}}|\mathbf{b}\rangle\langle \mathbf{b}|,
      \label{eq:SQR}
\end{align}
where the vector $\mathbf{w} = (w_1, w_2, \dots, w_n)$ defines the effective rotation angles $w_j$. When these operations are restricted to the $k$th level of the diagonal Clifford hierarchy, the phase weights satisfy the quantization condition:
\begin{equation}
    2^k\mathbf{w} \in \{0, 1, \dots, 2^k-1\}^n.
\end{equation}
Equivalently, this demonstrates that the phase polynomial is strictly linear, which provides the key structural simplification for analytically tracking the $N$-layer ansatz, yet it simultaneously imposes its most severe architectural restriction.

In the following sections, we exploit this mathematical structure to derive exact expressions for the dynamically evolving Pauli expectation values. These exact formulations subsequently enable us to establish rigorous bounds on quantum magic generation for both a shallow-layer model (ansatz~1) and an $N$-Layer model (ansatz~2).

\section{Quantum Magic and Circuit Ansatz in Early FTQC}
\label{sec:bg_nisq_ftqc}

Most previous analyses of quantum magic were developed in the context of the fully fault-tolerant~(FTQC) limit but required an ensemble average. In theoretical studies, T-doped random circuits are widely explored as it allows clear analysis thanks to the randomness~\cite{Leone_2021,Garcia_2023,Ahmadi_2024}. However, this is not suitable for analyzing a deterministic quantum circuit.  Consequently, relying on such ensemble-based analyses is practically infeasible for near-term hardware and provides limited insight into the actual magic-generating capabilities of specific architectural designs.

This conceptual shift is necessitated by computational models of early FTQC architectures, which fundamentally change logical circuit structure. In early FTQC, Clifford operations are comparatively cheap and can be implemented with high fidelity, whereas non-Clifford operations remain the dominant resource bottleneck. Logical circuits therefore naturally take the form of alternating Clifford layers and non-Clifford layers. The corresponding shift in perspective is summarized schematically in Fig.~\ref{Fig:2}. A particularly important example is the space-time efficient analog rotation (STAR) architecture. By utilizing specific subsystem codes, such as the $[[4,1,1,2]]$ code, these architectures implement arbitrary-angle Pauli-$Z$ rotations whose dominant residual errors are strongly biased toward logical $Z$ channels~\cite{PRXQuantum.5.010337,ismail2025transversalstararchitecturemegaquopscale,chung2026partiallyfaulttolerantquantumcomputation}. This makes diagonal non-Clifford gates, rather than random circuit ensembles, the natural objects of study in the early FTQC setting.

These architectures motivate the central optimization problem of this work: given an early FTQC circuit composed of alternating Clifford layers $\{C_i\}$ and diagonal non-Clifford gates $\{U_i\}$, how should one select these operations to maximize the quantum magic of the output state? In this context, \emph{maximum magic} represents more than a formal extremum; it serves as a natural benchmark for evaluating how effectively a constrained architecture converts costly non-Clifford resources into non-stabilizerness and, by extension, into classical simulation hardness~\cite{Howard_2017, Bravyi_2019}. This objective also highlights the intuitive appeal of a hierarchy-based heuristic~\cite{Beverland_2020}: if higher hierarchy levels strictly implied larger magic, then circuit design could be reduced to a simple algebraic ordering. However, as we show in the following sections, this is not the case. Optimal magic generation is intrinsically state-dependent and is fundamentally dictated by the structure of the Pauli spectrum.

\begin{figure}[t]
    \begin{minipage}[b]{\linewidth}
        \includegraphics[width=\linewidth]{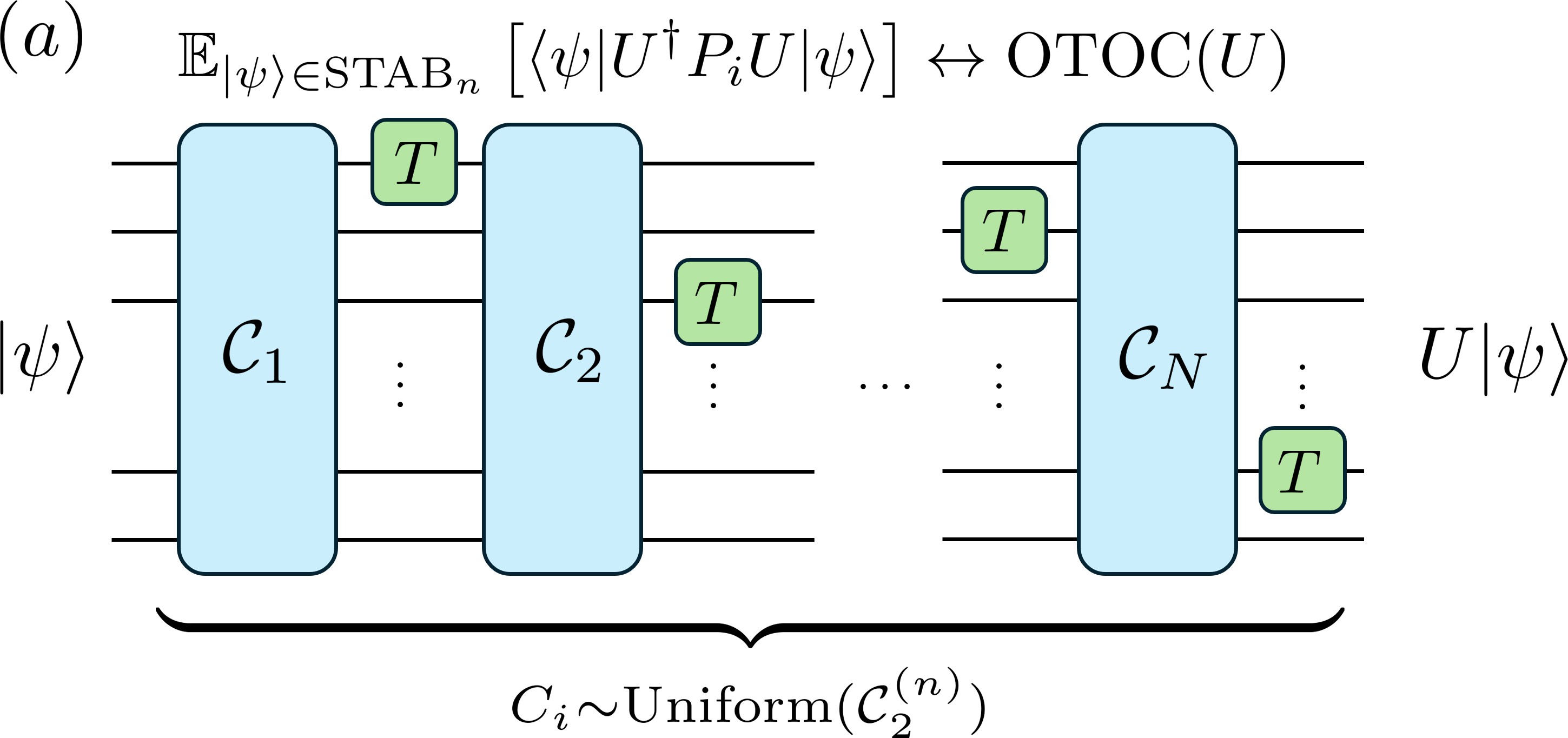}
    \end{minipage}%
    \\
    \vspace{2ex}
    \begin{minipage}[b]{\linewidth}
        \includegraphics[width=\linewidth]{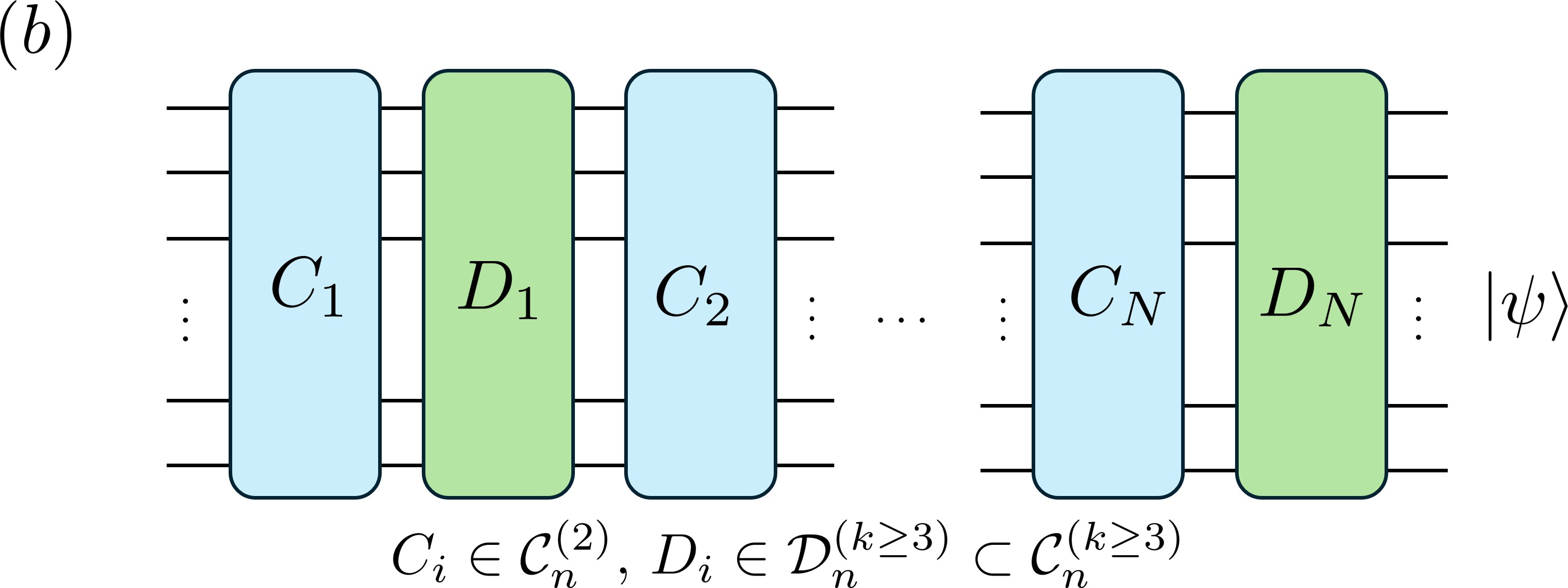}
    \end{minipage}
  \caption{Schematic comparison of circuit setups used to calculate quantum magic. (a) Standard theoretical paradigm based on deep $T$-doped random circuits. In these commonly analyzed models, non-stabilizerness is analytically evaluated by averaging over ensembles of uniformly distributed random Clifford blocks $C_i \sim \text{Uniform}(\mathcal{C}^{(2)}_n)$ to form unitary 3-designs that approximate Haar-random states~\cite{Leone_2021,Leone2022,Ahmadi_2024}. However, this framework relies on the perfect gate accuracy of mature fault-tolerant hardware. (b) A deterministic early FTQC framework analyzed this work. Computations are initialized in a given logical stabilizer state $|\psi\rangle$ and evolve through alternating layers of ideal Clifford operations $C_i \in \mathcal{C}_n^{(2)}$ and arbitrary diagonal non-Clifford gates $D_i \in \mathcal{D}_n^{(k \ge 3)}$ from the $k$th level of the Clifford hierarchy. This structure corresponds to hardware constraints of early FTQC designs, bypassing the need for ensemble averaging.}
  \label{Fig:2}
\end{figure}

Motivated by these considerations, we study the two representative early FTQC settings shown in Fig.~\ref{Fig:1}. In the \textbf{shallow-layer model} (ansatz~1), an initial stabilizer state evolves through a single Clifford layer followed by one layer of general diagonal-hierarchy gates. This model isolates the algebraic expressibility of diagonal operations and admits exact analysis of the resulting Pauli spectrum. In the \textbf{$N$-layer model} (ansatz~2), the circuit is composed of Clifford operations interleaved with single-qubit Pauli-$Z$ rotations. This restricted model captures a hardware-relevant subset of early FTQCs that are motivated by the STAR architecture and allows us to study how magic propagates under repeated Clifford mixing within the kinematic constraints of STAR-type implementations. Together, these two models separate general structural bounds from hardware-motivated dynamical constraints.

In the next section, we formalize this optimization problem precisely by establishing a Pauli-spectrum-based framework for experimentally accessible magic measures, thereby rendering the search for maximal magic analytically tractable.

\section{Main Result I: Structural Framework for Pure-State Magic Measures}
\label{sec:magic_category}
To optimize magic generation in early FTQC architectures, we need to choose an objective function. Although there have been extensive efforts to define magic measures satisfying natural properties~\cite{Howard_2017,Bravyi_2019,Beverland_2020,Leone2022}, focusing on a single specific measure is not desirable for our purposes. To make our optimization framework more general, this section presents two key results. 

In Sec.~\ref{subsec:classification}, we define two families of magic measures: the $L^p$ family and the Pauli-spectrum family. We then show that these families provide two useful structural frameworks for organizing broad classes of pure-state magic quantifiers. In particular, the \(L^p\) family gives stabilizer-expansion gauges, while the Pauli-spectrum family captures experimentally accessible measures that depend only on Pauli expectation values.

In Sec.~\ref{subsec:pauli_spectrum}, we investigate a characteristic structure common to functions in the Pauli-spectrum family. We show that maximizing any function in the Pauli-spectrum family can be rephrased as minimizing a Rényi-type polynomial of Pauli expectation values. This simplification provides a convenient framework for optimizing magic generation.

\subsection{Partial Classification of Magic Measures}
\label{subsec:classification}

\subsubsection{Natural properties for magic measures}
This paper focuses on magic measures for pure quantum states. This focus is justified in early FTQC regimes, where logical error rates are suppressed by quantum error correction. Let $Q$ be a function that maps a pure quantum state to a non-negative real value. We consider the following three conditions as natural properties required of a magic measure:
\begin{enumerate}
\item Faithfulness: $Q(\ket{\psi}) = 0$ if and only if $\ket{\psi} \in \mathrm{STAB}_n$, where $\mathrm{STAB}_n$ denotes the discrete set of all pure $n$-qubit stabilizer states.
\item Monotonicity: For any Clifford unitary or Pauli projection operator $\Lambda$, $Q(\Lambda \ket{\psi}) \leq Q(\ket{\psi})$.
\item Tensor-product subadditivity: $Q(\ket{\psi} \otimes \ket{\phi}) \leq Q(\ket{\psi}) + Q(\ket{\phi})$.
\end{enumerate}

Here, the faithfulness condition means that the magic measure is nonzero only for non-stabilizer states. The monotonicity condition means that the amount of magic never increases under Clifford operations or Pauli measurements. Tensor-product subadditivity means that the amount of magic cannot be increased by composing independent quantum states. It is therefore natural to require these properties of a function $Q$ that is intended to quantify magic.

\subsubsection{Families of magic measure}
Next, we define two families of magic measures. Both satisfy the properties discussed above and contain several widely used measures of quantum magic.
The first family is the $L^p$ family, which is defined in terms of stabilizer expansions. A stabilizer expansion expresses a pure state $\ket{\psi}$ as a linear combination of stabilizer states:
\begin{equation}
\ket{\psi}=\sum_i c_i \ket{s_i},\quad\text{where } \ket{s_i}\in\mathrm{STAB}_n.
\end{equation}
For a non-negative real number $p$, we define the characteristic function $\gamma_p$ as
\begin{equation}
\gamma_p(\ket{\psi}):=\inf\Bigl\{\|c\|_p:\ \ket{\psi}=\sum_i c_i\ket{s_i},\ \ket{s_i}\in\mathrm{STAB}_n\Bigr\},
\end{equation}
where $\|c\|_p=(\sum_i |c_i|^p)^{1/p}$. A magic measure in the $L^p$ family is then defined as a function of $\gamma_p$. For example, $\gamma_1^2$ is known as the stabilizer extent \cite{}. This family also includes several known magic measures, such as stabilizer-rank-related and convex-geometric measures.

The second family is the Pauli-spectrum family. The Pauli spectrum is the list of expectation values of Pauli operators. We denote a Pauli operator in symplectic form by
\begin{equation}
P(\mathbf{x},\mathbf{z}) = i^{\mathbf{x}\cdot\mathbf{z}} X^{\mathbf{x}} Z^{\mathbf{z}},
\end{equation}
where $\mathbf{x}, \mathbf{z} \in \mathbb{F}_2^{n}$ are binary vectors. The Pauli spectrum of $\ket{\psi}$ is defined as $\{ a_{\mathbf{x},\mathbf{z}}(\ket{\psi})\}_{\mathbf{x}, \mathbf{z}}$ where 
\begin{equation}
a_{\mathbf{x},\mathbf{z}}(\ket{\psi}) := \bra{\psi}P(\mathbf{x},\mathbf{z})\ket{\psi}.
\end{equation}
If a magic measure is defined as a function of the Pauli spectrum, we say that it belongs to the Pauli-spectrum family. For example, the stabilizer $\alpha$-Rényi entropy~\cite{Leone2022} belongs to this family and is defined as a moment of 
\begin{equation}
F_\alpha(\ket{\psi})=
\sum_{(\mathbf{x},\mathbf{z})\in \mathbb{F}_2^{2n}}
|a_{\mathbf{x},\mathbf{z}}(\ket{\psi})|^{2\alpha},
\end{equation}
where $\alpha$ is a positive integer. This family also includes several known magic measures, such as stabilizer nullity~\cite{Beverland_2020,xiao2026exponentially} and mana.

\subsubsection{Magic measure classification} 
We now explain in what sense the two families above provide a partial
structural classification of pure-state magic measures. The first result
shows that, under stabilizer normalization, positive homogeneity, and
subadditivity assumptions \footnote{
Here subadditivity is understood as a property of a stabilizer-decomposition
cost, rather than as a statement about physical magic generation by
superposition. If \(|\phi\rangle\) and \(|\psi\rangle\) are decomposed separately
into stabilizer states, then concatenating the two decompositions gives a
valid stabilizer decomposition of \(|\phi\rangle+|\psi\rangle\), which naturally
leads to \(Q(|\phi\rangle+|\psi\rangle)\le Q(|\phi\rangle)+Q(|\psi\rangle)\).
}, any admissible stabilizer-decomposition gauge is dominated by the stabilizer
\(L^1\) gauge.
\begin{theorem}[Stabilizer \(L^1\) gauge]
\label{thm:lp_family}
Let \(Q\) be a pure-state magic functional on state vectors that is
stabilizer-normalized, positively homogeneous, and subadditive:
\[
Q(|\phi\rangle+|\psi\rangle)
\le
Q(|\phi\rangle)+Q(|\psi\rangle).
\]
Then
\[
Q(|\psi\rangle)\le \gamma_1(|\psi\rangle).
\]
Moreover, \(\gamma_1\) is Clifford invariant and tensor-product
submultiplicative:
\[
\gamma_1(|\psi\rangle\otimes|\phi\rangle)
\le
\gamma_1(|\psi\rangle)\gamma_1(|\phi\rangle).
\]
\end{theorem}
A corresponding \(L^p\)-type hierarchy, obtained by replacing convexity with a \(p\)-decomposition inequality, is discussed in Appendix~\ref{app:Lp_stabilizer_family}.

\begin{proof} See Appendix \ref{app:Lp_stabilizer_family} \end{proof}
This theorem should be interpreted as a domination result rather than a
complete classification theorem. It shows that \(\gamma_1\) is the
maximal stabilizer-normalized, positively homogeneous, and subadditive
gauge in the stabilizer-expansion framework. Therefore, \(\gamma_1\)
provides a canonical benchmark for stabilizer-expansion-based measures. However, measures in the \(L^p\) family are not ideal for the optimization problems considered in this work, because evaluating them typically requires an optimization over an exponentially large discrete set of stabilizer states~\cite{Bravyi_2019,heimendahl2021stabilizer,PhysRevX.6.021043}. 
This motivates us to restrict magic measures to more computationally manageable measures, such as the Pauli spectrum family, at the cost of stronger assumptions.

To this end, we introduce the following additional requirements:
\begin{enumerate}
\item Permutation invariance: Let $\{a_{\mathbf{x}, \mathbf{z}}\}$ be the Pauli spectrum of a given state $\ket{\psi}$. The function $Q$ is permutation invariant if $Q(\ket{\psi})$ is invariant under any permutation of the Pauli spectrum.
\item Tensor-product additivity: $Q(\ket{\psi} \otimes \ket{\phi}) = Q(\ket{\psi}) + Q(\ket{\phi})$.
\item Pauli-spectrum continuity: Let $\{a_{\mathbf{x}, \mathbf{z}}\}$ be the Pauli spectrum of a given state $\ket{\psi}$. The function $Q$ is Pauli-spectrum continuous if $Q(\ket{\psi})$ is continuous with respect to changes in the Pauli spectrum.
\end{enumerate}
Here, permutation invariance is a stronger condition than monotonicity under Clifford unitary operations. Since the action of a Clifford unitary corresponds to a permutation of the Pauli spectrum, permutation invariance is a sufficient condition for monotonicity under Clifford operations. On the other hand, not every permutation of the Pauli spectrum can be realized by a Clifford operation, and therefore this assumption restricts the class of magic measures more strongly. Tensor-product additivity is a stronger variant of tensor-product subadditivity. Pauli-spectrum continuity is introduced to ensure that the functional can be represented as a continuous function of the Pauli spectrum.

Under these stronger assumptions, we obtain the following classification result.
\begin{theorem}[Unification with the Pauli-spectrum family]
\label{thm:pauli_spectrum_family}
Let $Q$ be a pure-state magic functional satisfying faithfulness, permutation invariance, tensor-product additivity, and Pauli-spectrum continuity. Then, $Q$ belongs to the Pauli-spectrum family.
\end{theorem}
\begin{proof} See Appendix \ref{app:pauli_spectrum_classification} \end{proof}

Therefore, by imposing these stricter conditions, we can convert the discussion of magic generation into one based on the Pauli-spectrum family. This family is particularly well suited to our analysis, because its elements can be inferred from Pauli expectation data and admit efficient sampling strategies~\cite{xiao2026exponentially,bittel2025operationalinterpretationstabilizerentropy}. Also, its analytical form is preserved under the circuit transformations of our interest, as discussed later in Secs.~\ref{sec:shallow_layer} and~\ref{sec:deep_layer}. 
The structural classification is summarized schematically in Fig.~\ref{Fig:3}.

\begin{figure}[t]
        \includegraphics[width=\linewidth]{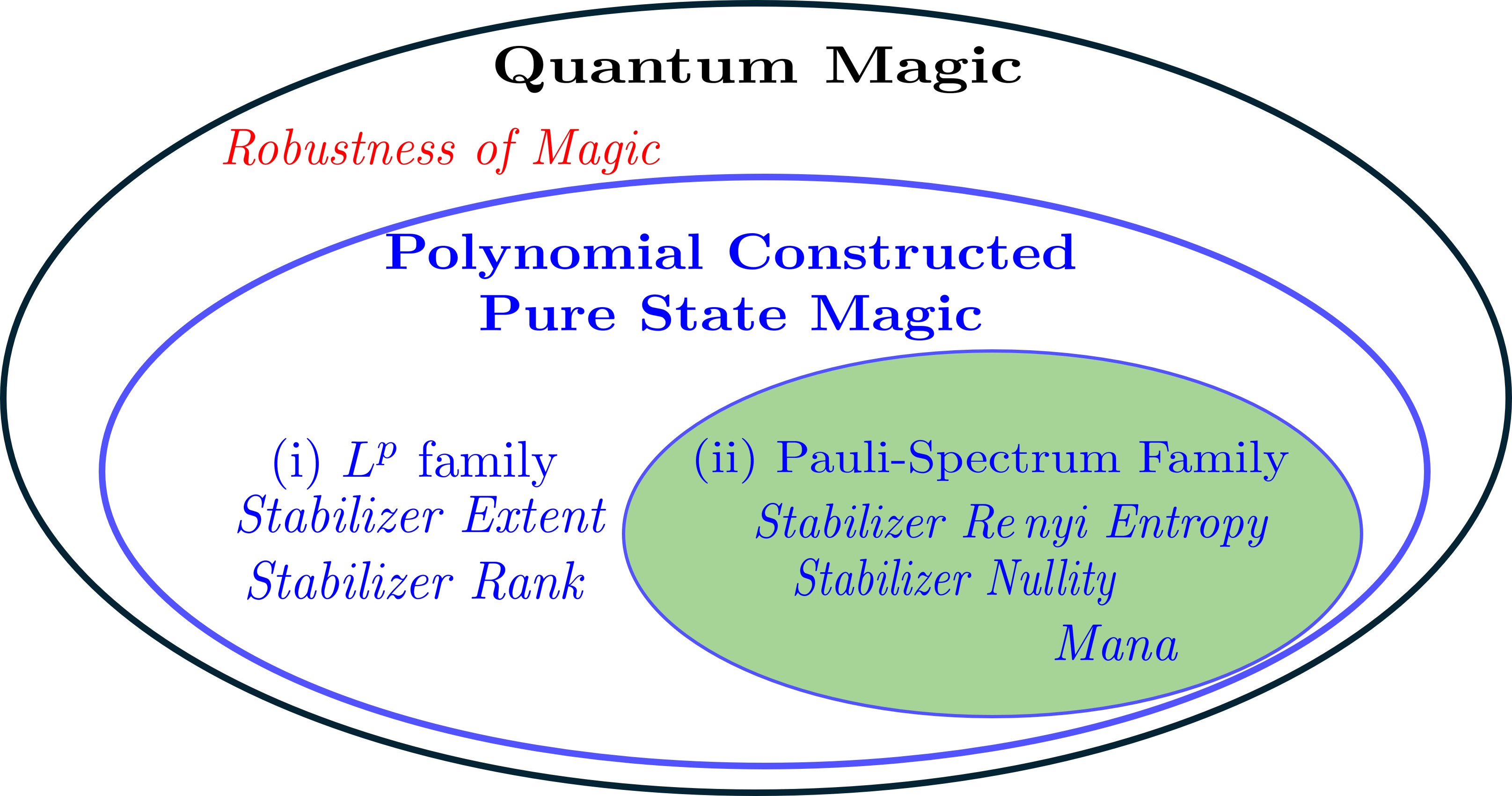}
  \caption{
Structural organization of the pure-state magic measures considered in this work. We highlight two useful operational families: 
(i) stabilizer-expansion measures, represented by the \(L^p\) hierarchy, and 
(ii) Pauli-spectrum measures, represented by functions of Pauli expectation values. 
The \(L^p\) family provides a natural convex-geometric benchmark, but is generally difficult to evaluate because it involves optimization over stabilizer decompositions. 
The Pauli-spectrum family contains quantities such as stabilizer Rényi entropy and stabilizer nullity, and is particularly suitable for early-FTQC resource tracking because it can be related to Pauli measurement data and admits tractable analytical propagation under the circuit transformations studied in this work.
}
  \label{Fig:3}
\end{figure}

\subsection{Characterization of Pauli-spectrum family}
\label{subsec:pauli_spectrum}
While the Pauli-spectrum family is of practical importance because of its direct experimental accessibility, there can be many detailed definitions in the family, which are not convenient for discussing the maximization of magic generation. Fortunately, we can further claim that any magic measure in the Pauli-spectrum family must take a Rényi-type form, and we can focus on a specific function to discuss the whole Pauli-spectrum family of magic measures. We formalize this property through the following uniqueness theorem.

\begin{theorem}[Uniqueness of Pauli-spectrum magic measures]
\label{thm:alpha_renyi}
Suppose $Q(|\psi\rangle)$ is a magic measure that belongs to a Pauli-spectrum family and satisfies permutation invariance, faithfulness, tensor-product additivity, and Pauli-spectrum continuity. Then it must take the Rényi-type form
\begin{equation}
Q(|\psi\rangle)= q\log d-\sum_{\alpha\ge2} k_\alpha \log F_\alpha(|\psi\rangle),
\label{eq:alpha_renyi_function}
\end{equation}
where $d=2^n$, the coefficients $k_\alpha>0$, $\alpha$ is integer, $q$ is the normalization factor to satisfy faithfulness, and
\begin{equation}
F_\alpha(|\psi\rangle)=\sum_{(\mathbf{x},\mathbf{z})\in \mathbb{F}_2^{2n}} |\langle\psi|P(\mathbf x, \mathbf z)|\psi\rangle|^{2\alpha}.
\label{eq:F_alpha}
\end{equation}
\end{theorem}

\begin{proof}[Proof sketch]
Since $Q$ depends only on the Pauli expectation values and is symmetric under their permutations, it must be constructed from permutation-invariant functions of the Pauli spectrum. Additivity then forces these functions to be built from multiplicative spectral moments, whose logarithms are additive. A natural closed set of multiplicative building blocks is given by the Rényi-type Pauli moments in Eq.~\eqref{eq:F_alpha}. Thus, any admissible functional in this class must take the form of a linear combination of $-\log F_\alpha$, up to the dimension-dependent offset $q\log d$ to make it vanish on the stabilizer state. This characterization is already implied at the level of Clifford-invariant Pauli-spectrum functionals. Since $Q$ is a positive linear combination of these Rényi-type moments. When the coefficients are chosen so that the corresponding Rényi-type quantities are magic monotones, known monotonicity results~\cite{Leone_2024} imply that their positive linear combinations also satisfy the required monotonicity property. The full proof is presented in Appendix~\ref{app:pauli_spectrum_classification}.
\end{proof}

Theorem~\ref{thm:alpha_renyi} is crucial for our analysis. It should be viewed not as a universal statement about all conceivable magic monotones, but as the tool that makes the rest of the paper analytically tractable. It shows that, within the moment-based Pauli-spectrum framework adopted here, the relevant additive pure-state magic measures can be organized in terms of the same basic moments \(F_\alpha\). Consequently, bounds or no-go results proved later for \(F_\alpha\) can be transferred to this operational family whenever the measure is expressed as a monotone function of these moments.

In the remainder of this work, we adopt the Pauli-spectrum family as the operational resource framework for early FTQC. The next section uses this framework in the shallow-layer model~(ansatz~1), where we derive the exact Pauli spectrum after a diagonal hierarchy gate acts on an arbitrary stabilizer input and convert that result into rigorous bounds on magic generation. We then extend the same framework to the $N$-layer model~(ansatz~2), where the evolving Pauli spectrum becomes the natural state-dependent object for optimization.

\section{Main Result II: Shallow Layer Model}
\label{sec:shallow_layer}

With the operational family of magic measures fixed in Sec.~\ref{sec:magic_category}, we can now return to the circuit-design problem posed in the Introduction: how much magic can a constrained early FTQC circuit generate, and what determines the optimum? The simplest nontrivial setting is
a single diagonal non-Clifford layer acting on a stabilizer
input. Structurally, this ansatz can be viewed as a generalization of Instantaneous Quantum Polynomial-time~(IQP) circuits~\cite{Shepherd_2009}, which cannot be efficiently simulated classically under standard complexity-theoretic assumptions~\cite{PhysRevLett.117.080501}. This shallow-layer regime already captures the essential structural tension of early fault tolerance: the diagonal layer is costly, but its ability to generate magic is not determined solely by the gate.

In practical early FTQC circuits, the computation is typically initialized in $|0\rangle^{\otimes n}$, after which a Clifford layer prepares the logical input to the non-Clifford stage~\cite{Litinski_2019}. Since Clifford operations map stabilizer states to stabilizer states, this preparation step can generate an arbitrary stabilizer state. Therefore, the shallow-layer ansatz
\begin{equation}
    |0\rangle^{\otimes n}
    \xrightarrow{\,C\,}
    |\psi\rangle \in \mathrm{STAB}_n
    \xrightarrow{\,U\,}
    |\psi_{\mathrm{out}}\rangle
\end{equation}
is mathematically equivalent to studying an arbitrary stabilizer input $|\psi\rangle \in \mathrm{STAB}_n$ acted on by a diagonal hierarchy gate $U \in \mathcal{D}_n^{(k)}$.
This viewpoint is especially natural for architecture-level analysis, since it separates the freely configurable Clifford preparation from the costly diagonal resource layer.

\begin{figure}[t]
    \includegraphics[width=\linewidth]{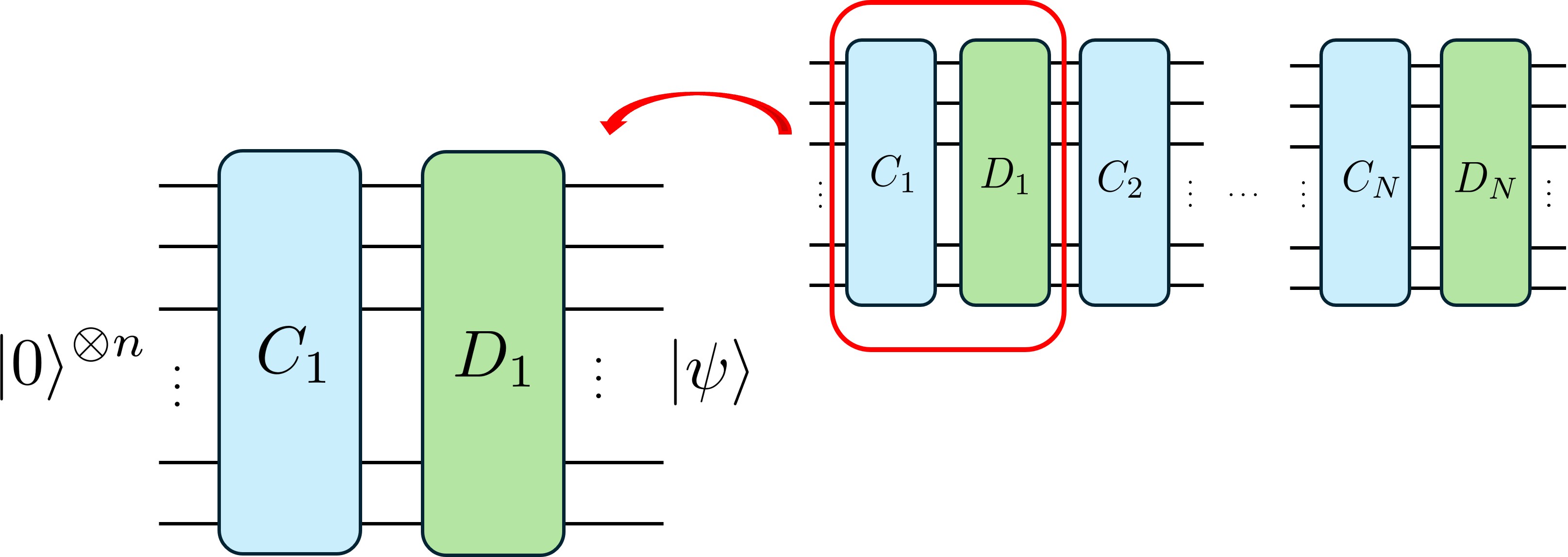}
  \caption{Architectural models for early fault-tolerant quantum computing (FTQC). A shallow-layer architectural model utilizing generalized operations from the diagonal Clifford hierarchy. Blue boxes $C_i$ denote Clifford gates; Green boxes denote diagonal hierarchy gates. The shallow-layer model permits arbitrary, multi-qubit general diagonal hierarchy gates.}
  \label{Fig:4}
\end{figure}

Here, we show that we can derive a closed-form Pauli spectrum. The purpose of the following derivation is not merely to obtain a closed-form representation, but to identify the structural variables that control magic generation in the shallow-layer model. In this setting, the key competition is between the phase interference introduced by the diagonal hierarchy gate and the support constraints imposed by the input stabilizer state. As we show below, the relevant stabilizer parameter is the number $r$ of independent pure-$Z$ generators, which directly limits how broadly the Pauli spectrum can spread.

Following standard stabilizer formalism, we represent the input stabilizer state $|\psi\rangle\in \text{STAB}_n$ by a phase vector $\mathbf h \in \mathbb{Z}_2^n$ and a stabilizer generator matrix $G$. After Gaussian elimination, $G$ can be brought to the form~\cite{gottesman1997stabilizercodesquantumerror}
\begin{equation}
    G=
    \begin{pmatrix}
        X_M & Z_M \\
        0   & Z_{\mathrm{pure}}
    \end{pmatrix},
\end{equation}
where the $r\times n$ submatrix $Z$ consists of $r$ linearly independent pure-$Z$ generators, defining the stabilizer group $\mathcal{S}$.

For a Pauli string $P(\mathbf{x},\mathbf{z})$ in symplectic form, define
\begin{equation}
    a_{P(\mathbf{x},\mathbf{z})}
    :=
    \langle \psi|U^\dagger P(\mathbf{x},\mathbf{z})U|\psi\rangle .
\end{equation}
Using the phase-polynomial representation of $U\in\mathcal{D}_n^{(k)}$, this expectation value through Eq.~\eqref{eq:diagonal_hierarchy} can be written as
\begin{align}
    a_{P(\mathbf{x},\mathbf{z})}
    &=
    \sum_{\mathbf{b}\in\mathbb{Z}_2^n}
    \exp\!\left(
        2\pi i \sum_m^{m'} \frac{\theta_m(\mathbf{b},\mathbf{x})}{2^m}
    \right)
    \langle \mathbf{b}\oplus\mathbf{x}|\psi\rangle
    \langle \psi|\mathbf{b}\rangle \notag\\
    &\qquad\times
    i^{\mathbf{x}\cdot\mathbf{z}}
    (-1)^{\mathbf{z}\cdot\mathbf{b}},
\end{align}
where $m'$ is the maximum phase resolution, $f_m(b)$ is the phase polynomial for diagonal Hierarchy and
\begin{equation}
    \theta_m(\mathbf{b},\mathbf{x})
    =
    f_m(\mathbf{b})-f_m(\mathbf{b}\oplus\mathbf{x}).
\end{equation}

This expression immediately makes the physical mechanism transparent. Since diagonal gates act only by attaching phases in the computational basis, magic can be generated only through interference between basis amplitudes already present in the stabilizer input. The quantity $ \langle \mathbf{b}\oplus\mathbf{x}|\psi\rangle\langle \psi|\mathbf{b}\rangle$ therefore encodes the geometric constraint imposed by the input stabilizer state, while the exponential factor encodes the non-Clifford phase structure of the diagonal gate. The shallow-layer problem is precisely the interplay between these two ingredients.

\begin{figure}[t]
\centering
\includegraphics[width=\linewidth]{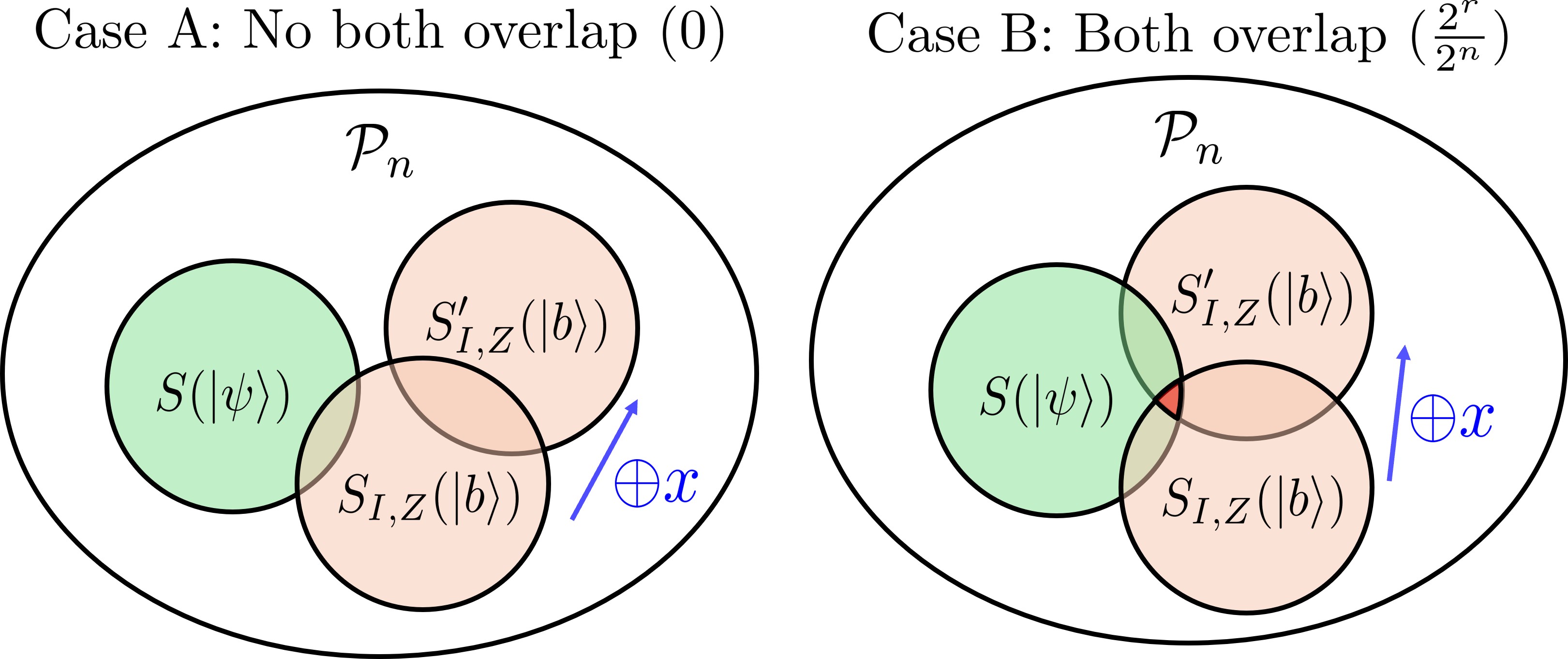}
\caption{Schematic diagram of the nontrivial group overlap, excluding the identity element. The red region denotes the simultaneous overlap between two computational-basis states and the stabilizer state, which determines the nonzero contributions to the Pauli spectrum.}
\label{Fig:5}
\end{figure}

Since computational-basis states $|\mathbf{b}\rangle$ are simultaneous eigenstates of pure-$Z$ stabilizers, the configurations $|\mathbf{b}\rangle$ and $|\mathbf{b}\oplus\mathbf{x}\rangle$ share the same generator matrix while differing only in their phase vectors. Consequently, the factor $\langle \mathbf{b}\oplus\mathbf{x}|\psi\rangle\langle \psi|\mathbf{b}\rangle$ directly measures the non-trivial overlap~(not including $I$) between the stabilizer group $\mathcal{S}$ and the corresponding computational-basis components. As illustrated schematically in Fig.~\ref{Fig:5}, this overlap is nonzero only when the two configurations are compatible with the same stabilizer structure.

This overlap can be evaluated explicitly in terms of the $r$ pure-$Z$ generators of $\mathcal{S}$. For $\mathbf{x}\in\mathbb{L}_x$, one finds
\begin{equation}
    \left|
    \langle \mathbf{b}\oplus\mathbf{x}|\psi\rangle
    \langle \psi|\mathbf{b}\rangle
    \right|
    =
    \begin{cases}
        \dfrac{2^r}{2^n},
        & \text{if }\mathbf{h}'\oplus(\mathbf{b} Z_{\mathrm{pure}})\oplus\mathbf{K}=\mathbf{0}^r,\\[6pt]
        0,
        & \text{otherwise},
    \end{cases}
\end{equation}
where $\mathbb{L}_x$ is the $(n-r)$-dimensional projected Lagrangian subspace generated by $X_M$, $\mathbf{K}$ is determined by
\begin{equation}
    2\mathbf{K}=\mathbf{x} Z_{\mathrm{pure}} \pmod 4,
\end{equation}
and $\mathbf{h}'\in\mathbb{Z}_2^r$ is the projected phase subvector associated with the pure-$Z$ generators. 

Substituting this condition yields the exact closed-form Pauli expectation value for $\mathbf{x}\in\mathbb{L}_x$:
\begin{align}
    a_{P(\mathbf{x},\mathbf{z})}
    &=
    \sum_{\mathbf{b}\in\mathbb{Z}_2^n}
    \exp\!\left(
        2\pi i \sum_{m=1}^{m'}
        \frac{\theta_m(\mathbf{b},\mathbf{x})}{2^m}
    \right)
    (-1)^{\mathbf{z}\cdot\mathbf{b}}
    \frac{2^r}{2^n}
    \notag\\
    &\times
    \bigl[
        (-1)^{s_0+\mathbf{x}\cdot\mathbf{z}_{\mathrm{ref}}+\mathbf{b}\cdot\mathbf{z}_{\mathrm{ref}}}
        i^{\mathbf{x}\cdot\mathbf{z}_{\mathrm{ref}}}
    \bigr]
    \notag\\
    &\times\delta\!\bigl(
        \mathbf{0}^r,\,
        \mathbf{h}'\oplus(\mathbf{b} Z_{\mathrm{pure}})\oplus\mathbf{K}
    \bigr),
    \label{eq:shallow_model_pauli_spectrum}
\end{align}
where $s_0$ and $\mathbf{z}_{\mathrm{ref}}$ correspond to the symplectic reference vector belongs to Lagrange subspace $(\mathbf{x},\mathbf{z}_{\mathrm{ref}})\in\mathbb{L}$. This expression automatically satisfies the normalization condition, remains transparent on pure-$Z$ Pauli strings, and reproduces the expected lower-level limits. Most importantly for our purposes, it makes explicit that magic generation depends jointly on the phase-polynomial structure of $U$ and on the stabilizer geometry encoded by $r$. The detailed calculation is presented in Appendix~\ref{app:shallow_analytic}.

Summing over the valid Pauli spectrum then yields the polynomial functional underlying the experimental operational $\alpha$-Rényi function $F_\alpha$ for all non-zero integer $\alpha$ is:
\begin{align}
        &F_{\alpha}(U|\psi\rangle)  \notag \\
    =& \sum_{\mathbf x\in \mathbb{L}_x} \sum_{\bigoplus\limits_{i=1}^\alpha \mathbf b'_i \oplus \mathbf b_i = 0^n} \hspace{-3ex}
\exp\!\left(
  2\pi i \sum_{m=1}^{m'} \sum_{i=1}^\alpha\frac{\theta_m(\mathbf b_i,\mathbf x) - \theta_m(\mathbf b_i',\mathbf x)}{2^m}
\right) \notag \\
& \frac{4^{\alpha r}}{4^{\alpha n}} 2^n
\delta\!\bigl(\mathbf 0^{r},\, \mathbf h' \oplus (\mathbf b_i Z_{\mathrm{pure}}) \oplus \mathbf K\bigr)\bigr)
\label{eq:magic polynominal}
\end{align}
The combinatorial effect of the phase polynomial $\theta_m$ inside this functional can be understood naturally through a Walsh--Hadamard viewpoint~\cite{park2026samplehardwareefficientfidelityestimation}.

The exact Pauli spectrum derived above therefore provides a complete description of all experimentally operational pure-state magic measures in this model, since by Theorem~\ref{thm:alpha_renyi} every such measure reduces to a function of the moments $F_\alpha$.

This result thus indicates not only that the shallow-layer model is analytically solvable, but also that the solution already exposes the design variables relevant for early FTQC. The diagonal layer controls the available phase interference, while the Clifford preparation controls the stabilizer geometry of the input. In particular, the integer $r$ enters as a direct bottleneck on Pauli-spectrum spreading. This observation is the starting point for the bounds derived in the following section for the lower and upper bounds of magic generation.

\subsection{Lower bound: zero magic remains possible}

To state that the lower bound including zero magic generation remains possible, let $r$ denote the maximal number of independent pure-$Z$ generators in the stabilizer group of the input state. Operationally, $r$ quantifies how restricted the support of the input stabilizer state is in the computational basis: each independent pure-$Z$ generator removes one binary degree of freedom from the allowed basis strings.

\begin{theorem}[Zero-magic bound in the shallow-layer model]
\label{thm:zero_magic_shallow}
Let $|\psi\rangle \in \mathrm{STAB}_n$ be an $n$-qubit stabilizer state with $r$ independent pure-$Z$ generators. Let $k$ be an integer that satisfies 
\begin{equation}
    k \le r + 2.
\end{equation}
Then there exists a genuinely level-$k$ diagonal gate
\begin{equation}
    U \in \mathcal D_n^{(k)} \setminus \mathcal D_n^{(k-1)}
\end{equation}
such that
\begin{equation}
    U|\psi\rangle \in \mathrm{STAB}_n.
\end{equation}
In particular, all operational pure-state magic measures vanish on the output.
\end{theorem}

\begin{proof}[Proof sketch]
The key point is that a diagonal hierarchy gate generates magic only through phase interference on the support of the input stabilizer state. If the input has $r$ independent pure-$Z$ generators, then after a suitable Clifford change of basis, it is equivalent to
\begin{equation}
    |0\rangle^{\otimes r} \otimes |+\rangle^{\otimes (n-r)}.
\end{equation}
So the diagonal phase polynomial is effectively evaluated only on the remaining $n-r$ unfrozen variables. When $k \le r + 2$, one can choose a genuinely level-$k$ phase polynomial whose higher-level contribution vanishes on this restricted support, leaving only a Clifford-compatible phase and therefore a stabilizer output. The full proof is presented in Appendix~\ref{app:shallow_bound}.
\end{proof}

This mechanism can also be applied from the opposite perspective. For \emph{any} fixed, genuinely level-$k$ diagonal gate
\begin{equation}
U \in \mathcal{D}_n^{(k)} \setminus \mathcal{D}_n^{(k-1)},
\end{equation}
provided the system size satisfies $n \ge k-1$, one can always find a nontrivial stabilizer input $|\psi\rangle$ (i.e., a state that is not a computational-basis state, implying $r \le n-1$) such that
\begin{equation}
U|\psi\rangle \in \mathrm{STAB}_n.
\end{equation}
Combining the zero-magic condition $r \ge k-2$ from Theorem~\ref{thm:zero_magic_shallow} with the non-triviality condition $r \le n-1$ mathematically guarantees that such a valid stabilizer geometry exists whenever $n \ge k-1$. Thus, the zero-magic phenomenon is not merely an artifact of specific input families, but a generic obstruction for any diagonal gate given an appropriately tailored input stabilizer state.

Theorem~\ref{thm:zero_magic_shallow} gives the absolute lower envelope of magic generation in the shallow-layer model. It shows that even genuinely higher-level diagonal gates can completely fail to generate magic. Consequently, the Clifford hierarchy level, taken in isolation, cannot serve as a guaranteed proxy for non-stabilizer resource generation.

\subsection{Upper bound: maximal magic requires graph-state geometry}

We now pose the converse question: which stabilizer inputs are best suited for converting a single diagonal layer into as much operational magic as possible? Since every operational measure in our framework is determined by the moments $F_\alpha$, this question reduces to minimizing $F_\alpha$ under the Pauli normalization constraint
\begin{equation}
    \sum_{(\mathbf{x},\mathbf{z})\in \mathbb{F}_2^{2n}} |a_{P(\mathbf{x},\mathbf{z})}|^2 = 2^n.
\end{equation}

\begin{theorem}[Flat minimum and graph-state condition]
\label{thm:graph_state_condition}
Let
\begin{equation}
    |\psi\rangle \in \mathrm{STAB}_n,
    \qquad
    U \in \mathcal D_n^{(k)} \setminus \mathcal D_n^{(k-1)},
\end{equation}
and let $r$ denote the number of independent pure-$Z$ generators of $|\psi\rangle$. If the output $U|\psi\rangle$ attains the minimum of $F_\alpha$ allowed by the shallow-layer ansatz, then the following conditions are necessary:

\begin{enumerate}
    \item The non-identity Pauli spectrum is as flat as the constraints allow. In other words, for all  $\mathbf{x} \neq \mathbf{0}$ satisfy 
    \begin{equation}
        |a_{P(\mathbf{x},\mathbf{z})}|^2
        =
        2^{-n}.
        \label{eq:shallow_pauli}
    \end{equation}

    \item The input stabilizer must satisfy
    \begin{equation}
        r = 0.
    \end{equation}
    Equivalently, the input has no pure-$Z$ stabilizers.
\end{enumerate}

Consequently, for every $\alpha \ge 2$, the magic generation is strictly bounded by
\begin{equation}
    F_\alpha(U|\psi\rangle)
    \ge
    1 + (2^n-1)2^{n(1-\alpha)},
    \label{eq:shallow_max_magic}
\end{equation}
with equality if and only if the accessible non-identity Pauli spectrum is exactly flat.
\end{theorem}

\begin{proof}[Proof sketch]
For fixed $\sum |a_P|^2 = 2^n$, the quantity $F_\alpha = \sum |a_P|^{2\alpha}$ is minimized when the Pauli weight is distributed as uniformly as possible. The exact shallow-layer spectrum shows that pure-$Z$ generators impose Kronecker-delta constraints on the allowed support, so any $r>0$ forces part of the Pauli distribution to remain concentrated rather than spread. Hence, maximal flattening requires $r=0$. In that case, the identity contribution remains fixed, while the remaining weight $2^n-1$ must be distributed over the $4^n-2^n$ non-identity coefficients, which gives the flat value $2^{-n}$ and the stated lower bound on $F_\alpha$. The full proof is presented in Appendix~\ref{app:shallow_bound}.
\end{proof}

Theorem~\ref{thm:graph_state_condition} provides an actionable design principle for shallow early FTQC circuits: if one wants a single diagonal layer to generate as much operational magic as possible, then the preceding Clifford layer should first prepare the input into a stabilizer state with $r=0$. By definition, a stabilizer state that lacks pure-$Z$ generators corresponds uniquely to a graph state, up to local Clifford operations~\cite{Hein_2004}. In other words, the relevant optimization variable is not solely the non-Clifford gate, but also the specific structure of the stabilizer state immediately before it.

\subsection{No-go theorem: no hierarchy ordering by operational magic}

We can now connect the two preceding results. The lower bound shows that zero magic is always attainable at a fixed hierarchy level on suitable stabilizer inputs. The upper-bound analysis, in turn, shows that approaching maximal magic requires a highly specific input geometry, namely $r=0$ or, equivalently, a graph-state orbit. Thus, even before considering deeper circuits, the shallow-layer model already reveals that magic generation is controlled at least as much by input stabilizer structure as by hierarchy level itself. Since all operational pure-state magic measures reduce to the same family $\{F_\alpha\}$, no state-independent ordering principle based solely on hierarchy level can hold. This conceptual misalignment is illustrated schematically in Fig.~\ref{Fig:6}.

\begin{corollary}[No strict ordering of diagonal hierarchy levels]
\label{cor:no_go_shallow}
There exists no operational pure-state magic measure within the Pauli-spectrum family that can strictly order the diagonal Clifford hierarchy by level. Equivalently, for any $k \ge 3$, one cannot have a universal separation of the form
\begin{equation}
    Q\bigl(U_k |\psi\rangle\bigr)
    >
    Q\bigl(U_{k-1} |\psi\rangle\bigr),
\end{equation}
or
\begin{equation}
    Q\bigl(U_k |\psi\rangle\bigr)
    <
    Q\bigl(U_{k-1} |\psi\rangle\bigr),
\end{equation}
holding for arbitrary stabilizer inputs $|\psi\rangle$ and all choices
\begin{equation}
    U_k \in \mathcal D_n^{(k)} \setminus \mathcal D_n^{(k-1)},
    \qquad
    U_{k-1} \in \mathcal D_n^{(k-1)} \setminus \mathcal D_n^{(k-2)}.
\end{equation}
\end{corollary}

\begin{proof}[Proof]
The argument uses only the zero-magic bound in Theorem~\ref{thm:zero_magic_shallow} together with faithfulness of the Pauli-spectrum family. Suppose first that a universal strict ordering existed with
$Q\bigl(U_k|\psi\rangle\bigr) > Q\bigl(U_{k-1}|\psi\rangle\bigr)$
for arbitrary stabilizer input $|\psi\rangle$. Choosing $|\psi\rangle$ with $r\ge k-2$, Theorem~\ref{thm:zero_magic_shallow} guarantees a genuinely level-$k$ gate $U_k$ such that
$U_k|\psi\rangle \in \mathrm{STAB}_n$.
Hence $Q(U_k|\psi\rangle)=0$, which contradicts strict positivity over the adjacent lower level since every admissible $Q$ is nonnegative.

A reverse strict ordering is ruled out identically: choosing instead a stabilizer input with $r\ge k-3$, Theorem~\ref{thm:zero_magic_shallow} applied at level $k-1$ gives a genuinely level-$(k-1)$ gate $U_{k-1}$ with
$U_{k-1}|\psi\rangle \in \mathrm{STAB}_n,$
so again $Q(U_{k-1}|\psi\rangle)=0$, contradicting strict separation in the opposite direction. Therefore, no operational pure-state magic measure in the Pauli-spectrum family can induce a universal, state-independent strict ordering of adjacent diagonal hierarchy levels.
\end{proof}

\begin{figure}[t]
\centering
\includegraphics[width=\linewidth]{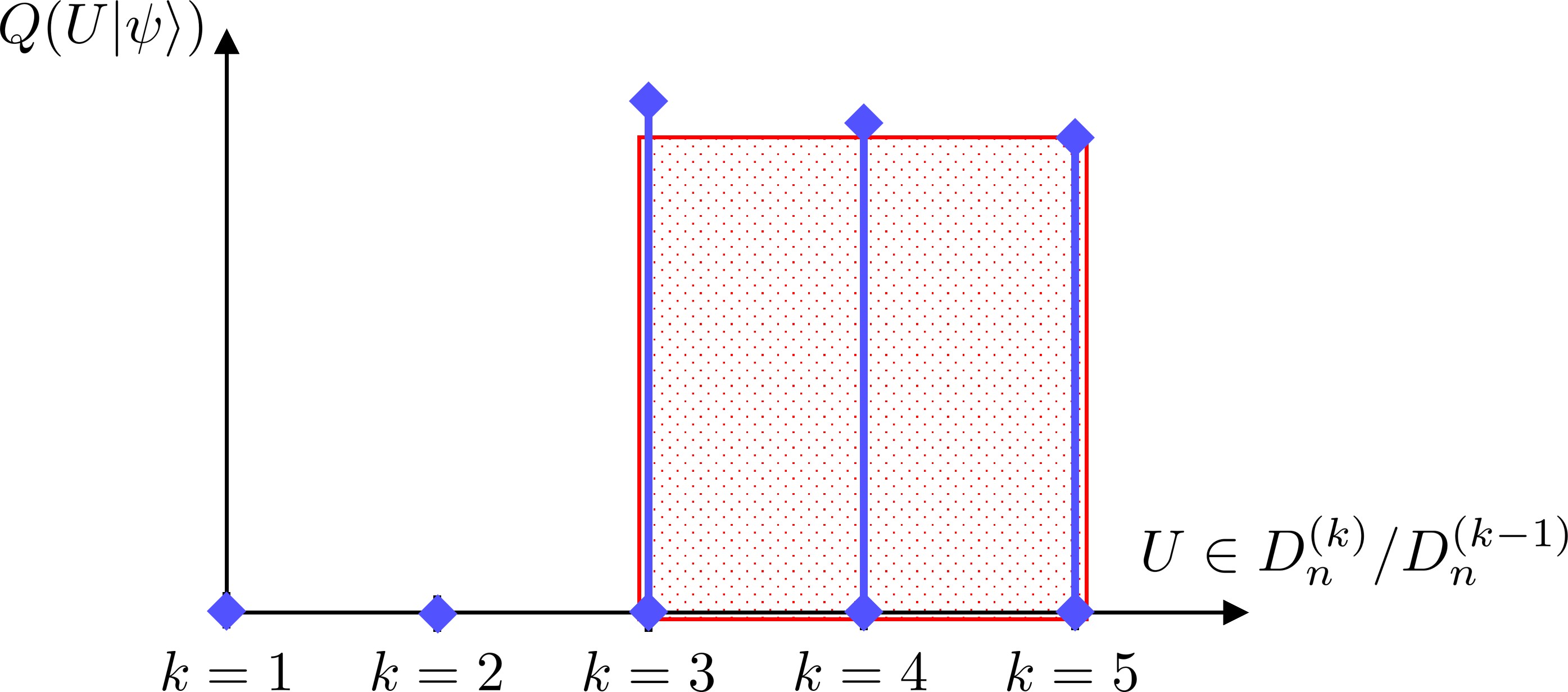}
  \caption{Schematic illustration of the no-go theorem. Owing to significant overlap between the upper and lower bounds, gates within the $D_n^{(k)}/D_n^{(k-1)}$ levels cannot be distinguished based on their quantum magic values. Since the diagonal hierarchy is a subgroup of the Clifford hierarchy, this operational indistinguishability extends to the Clifford hierarchy as well.}
  \label{Fig:6}
\end{figure}

This corollary serves as the central conclusion of the shallow-layer analysis. The diagonal Clifford hierarchy is an algebraic classification of gates, whereas quantum magic is a state-dependent geometric resource. In the shallow-layer regime, these two notions do not collapse into a single monotone ordering. Therefore, the relevant design lesson for early fault tolerance is not to choose the highest hierarchy level available, but rather to choose the gate and stabilizer geometry that allow it to maximally spread the Pauli spectrum.

\subsection{Implication for circuit design}

The shallow-layer model establishes a concrete, actionable guideline for early FTQC architecture: the configuration of the initial Clifford layer is itself a critical component of magic engineering. Rather than treating this initial stage as a trivial preparatory step, our bounds demonstrate that it fundamentally dictates the efficiency with which expensive non-Clifford operations are converted into operational magic.

From a pessimistic perspective, Theorem~\ref{thm:zero_magic_shallow}  highlights a severe architectural pitfall: it is physically possible to entirely squander an expensive diagonal non-Clifford layer. Despite being genuinely outside the Clifford group, a diagonal gate acting on a highly restricted input ($r \ge k-2$) can produce a stabilizer output, generating strictly zero operational magic. From a resource-tracking perspective, such a costly operation is rendered indistinguishable from a free Clifford gate. To circumvent this, gate selection $\{C_i,D_i\}$ must actively avoid choosing configurations containing large pure-$Z$ stabilizer sectors.

Conversely, Theorem~\ref{thm:graph_state_condition}  provides the optimal design target. Maximizing shallow-layer magic explicitly requires preconditioning the input into a graph state, up to local Clifford equivalence. This unique stabilizer geometry completely eliminates pure-$Z$ bottlenecks ($r=0$), enabling the Pauli spectrum to spread as uniformly as the algebraic constraints of the architecture permit. Operationally, this dictates that the optimal initial preparation step in a magic-generating circuit is a graph state.

Taken together, these conclusions show that the first layer in early FTQC cannot be treated merely as a preparatory step for later dynamics. It already sets the efficiency with which the first costly diagonal layer is converted into operational magic. As we will show in the next section for the $N$-layer model (ansatz~2), this initial choice continues to matter at later depths because subsequent layers act on the Pauli spectrum generated by earlier layers. The optimization problem is therefore state-dependent already at depth one, and hierarchy level alone does not provide a reliable design principle.

\section{Main Result III: \texorpdfstring{$N$}{N}-layer SQR Generalization}
\label{sec:deep_layer}
While the shallow-layer analysis in Sec.~\ref{sec:shallow_layer} establishes how the initial diagonal layer should be configured to maximize magic generation, we now transition to the more physically realistic early FTQC regime. In this setting, Clifford and non-Clifford gates are applied alternately across $N$ layers. The relevant question is therefore no longer solely how much magic a single layer can generate, but rather how this resource dynamically evolves throughout a restricted multi-layer architecture. To address this, we study the $N$-layer model (ansatz~2), which is directly motivated by early FTQC hardware constraints and the STAR architecture depicted in Fig.~\ref{Fig:7}. Crucially, this ansatz restricts the non-Clifford layers strictly to single-qubit $Z$-rotations (SQR). By excluding multi-qubit diagonal interactions, this model represents a highly constrained subset of the general diagonal framework analyzed in Sec.~\ref{sec:shallow_layer}.

\begin{figure}[b]
        \includegraphics[width=\linewidth]{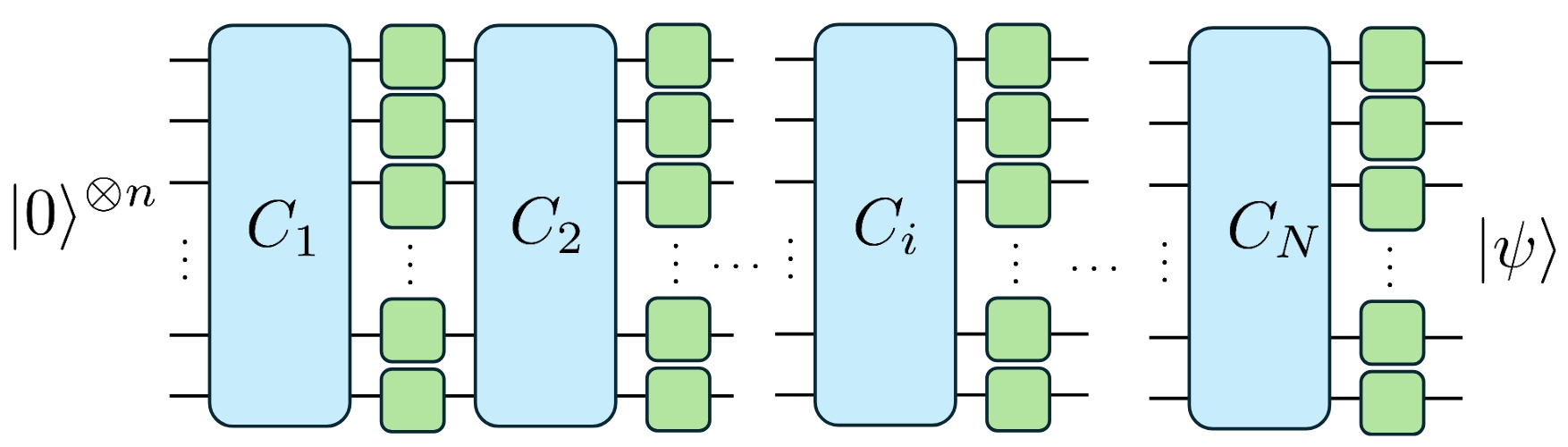}
    \caption{Architectural model for early fault-tolerant quantum computing (FTQC). The $N$-layer model (ansatz~2) alternates Clifford layers $C_i$ with native single-qubit diagonal $Z$-rotation layers. Blue boxes denote Clifford gates and green boxes denote diagonal hierarchy gates. In contrast to the shallow unrestricted model, the non-Clifford layer here is restricted to $n$ parallel single-qubit diagonal operations.}
    \label{Fig:7}
\end{figure}

For a single SQR layer acting on a stabilizer input, the exact Pauli spectrum simplifies from the general diagonal expression in Sec.~\ref{sec:shallow_layer} and SQR gates in Eq.~\eqref{eq:SQR} to
\begin{equation}
a_{P(\mathbf{x},\mathbf{z})}
=
\frac{2^r}{2^n}
C(\mathbf{x})
\sum_{\substack{\mathbf{b}\in\mathbb{Z}_2^n\\ \mathbf{b}\cdot Z=\mathbf{h}'\oplus \mathbf{K}}}
\exp\!\left(2\pi i\mathbf{V}(\mathbf{x},\mathbf{z})\cdot \mathbf{b}\right),
\label{eq:shallow_STAR_spectrum}
\end{equation}
where $C(x)$ is an overall phase term canceled later
and the effective phase vector is
\begin{equation}
    \mathbf{V}(\mathbf{x},\mathbf{z})
    =
    2(\mathbf{w}\circ \mathbf{x})
    +
    \frac{\mathbf{z}\oplus \mathbf{z}_{\mathrm{ref}}}{2}.
    \label{eq:star_phase_vector}
\end{equation}
with $\circ$ denoting the element-wise (Hadamard) product. The significance of Eq.~\eqref{eq:shallow_STAR_spectrum} is not merely that it gives a one-layer closed form, but that it explicitly delineates how STAR differs from the unrestricted shallow model: each analog layer can only induce a highly structured redistribution of Pauli weight, controlled by the rotation vector $\mathbf{w}$.

This localized structure makes the multilayer problem analytically tractable. A Clifford layer acts on Pauli operators by symplectic permutation and phase reweighting~\cite{gottesman_knill,Aaronson_2004}, while the subsequent analog $Z$-rotation layer mixes only a restricted subset of Pauli expectation values. Combining these two steps yields an exact iterative update rule for the Pauli spectrum. If $a_{P(x,z)}$ denotes the Pauli expectation values before the new layer, then after one additional Clifford-plus-rotation block, they become
\begin{align}
a'_{P(\mathbf{x},\mathbf{z})}
=&
\sum_{\mathbf{u}\le \mathbf{x}}
\Gamma_{\mathbf{x}}(\mathbf{u};\mathbf w)
(-1)^{|\mathbf{u}|+\mathbf{u}\cdot \mathbf{z}+h(F^{-1}(\mathbf{x},\mathbf{z}\oplus \mathbf{u}))} \notag \\
&\times a_{P(F^{-1}(\mathbf{x},\mathbf{z}\oplus \mathbf{u}))},
\label{eq:STAR_iterative_update}
\end{align}
where the Clifford layer acts via the symplectic transformation $F\in\mathrm{Sp}(2n, \mathbb{F}_2)$ and $h(\cdot)$ accounts for the phase factors induced by Clifford conjugation. The coefficients $\Gamma_x(u)$ are determined by the analog rotation angles $\mathbf{w}$:
\begin{equation}
     \Gamma_{\mathbf x}(\mathbf u;\mathbf w) = \prod_{j: x_j=1} \left( \cos(2\pi w_j)^{1-u_j} ( \sin(2\pi w_j))^{u_j} \right)
\end{equation}
Here, the sum runs over all binary vectors $\mathbf{u} \in \mathbb{Z}_2^n$ that are covered by the support of $\mathbf{x}$, denoted as $\mathbf{u} \le \mathbf{x}$, meaning $u_j = 0$ whenever $x_j = 0$.

This iterative structure in Eq.~\eqref{eq:STAR_iterative_update} forms the central analytical result for the ansatz~2. It reveals that each new layer acts as a restricted transfer rule on the existing Pauli expectation value: the Clifford part permutes and rephases the coefficients, while the analog rotation layer performs only a constrained local mixing. In particular, this demonstrates that magic generation is not an independent, additive process. Instead, the non-stabilizer resource yielded by any new layer is intrinsically state-dependent and strictly dictated by the specific Pauli distribution prepared by all preceding layers. By exposing this structural dependence, our formulation naturally motivates the next critical inquiry: Does there exist any universal, state-independent rule for selecting operations that can guarantee monotonic magic enhancement? As we establish in the following subsection, the answer is fundamentally negative.

\subsection{No-go theorem for state-independent layer choice}

Equation~\eqref{eq:STAR_iterative_update} shows that, in the $N$-layer model (ansatz~2), the Pauli spectrum after each additional layer is obtained from the previous one by a linear transformation. If we collect the Pauli expectation values into a real vector
\begin{equation}
    \mathbf{a}
    =
    \{a_{P(\mathbf x,\mathbf z)}\}
    \in \mathbb{R}^{4^n},
\end{equation}
then the update induced by one Clifford-plus-rotation block can be written abstractly as $\mathbf{a}\longmapsto \mathbf{a}'=M\mathbf{a}$.

Because the Pauli expectation values satisfy the normalization
\begin{equation}
    \|\mathbf{a}\|_2^2 = \sum_{(\mathbf x,\mathbf z) \in \mathbb{F}_2^{2n}}|a_{P(\mathbf x,\mathbf z)}|^2 = 2^n,
\end{equation}
the transfer matrix $M$ preserves the Euclidean norm and therefore acts orthogonally on the Pauli vector. In this form, the SQR layer update can be viewed as an orthogonal redistribution of Pauli weight.

This reformulation allows us to ask a natural architecture-level question: does there exist a state-independent rule for choosing the next layer, namely the next Clifford mixing and analog rotation parameters, such that magic always improves? Equivalently, can one choose a transfer map that universally flattens the Pauli spectrum, regardless of the current state? The following theorem shows that the answer is no.

\begin{theorem}[No-go for universal Pauli-spectrum flattening under orthogonal transfer]
\label{thm:no_go_universal_flattening}
Let $d=4^n$, and let $\mathbf{a}\in\mathbb{R}^d$ denote the Pauli-coefficient vector of a pure $n$-qubit state, normalized so that $\|\mathbf{a}\|_2^2 = 2^n$. Let $M\in O(d)$ be a real orthogonal matrix acting linearly on $\mathbf{a}$, $\mathbf{a}' = M\mathbf{a}$.
For any $\alpha>1$, define
\begin{equation}
    F_\alpha(\mathbf{a})
    :=
    \sum_{j=1}^d |a_j|^{2\alpha}
    =
    \|\mathbf{a}\|_{2\alpha}^{2\alpha}.
\end{equation}
Then there does not exist an orthogonal matrix $M$ such that
\begin{equation}
    F_\alpha(M\mathbf{a}) < F_\alpha(\mathbf{a})
\end{equation}
for all pure-state Pauli vectors $\mathbf{a}$ with $\|\mathbf{a}\|_2^2=2^n$, with strict inequality on a set of nonzero measure. Equivalently, an orthogonal transformation cannot universally flatten the Pauli spectrum: if it decreases $F_\alpha$ for some states, it must increase $F_\alpha$ for others.
\end{theorem}

\begin{proof}[Proof sketch]
The claim follows from the geometry of $\ell_p$ norms on a fixed $\ell_2$ sphere. Since $M$ is orthogonal, it preserves $\|\mathbf{a}\|_2$ and therefore only redistributes weight among the Pauli coefficients. For $\alpha>1$, the functional $F_\alpha(\mathbf{a})=\|\mathbf{a}\|_{2\alpha}^{2\alpha}$ is minimized by flatter distributions and increased by more concentrated ones. If a single orthogonal map $M$ were able to decrease $F_\alpha$ for every Pauli vector, then $M^{-1}$ would also be orthogonal and would have to reverse this flattening on the image set, producing a contradiction. Equivalently, orthogonal transformations can reshuffle Pauli weight, but no fixed orthogonal reshuffling can move all states monotonically toward flatter spectra. The full proof is presented in Appendix~\ref{app:no_go_2}.
\end{proof}

Theorem~\ref{thm:no_go_universal_flattening} imposes a direct physical constraint on multilayer SQR design. It proves that there is no state-independent prescription for choosing the next Clifford layer or the next set of analog rotation angles that guarantees monotonic magic enhancement for every input spectrum. In particular, no universal rule of the form \textbf{always apply specific Clifford mixing} or \textbf{always choose certain rotation angles} can succeed throughout the evolution. Even at the level of the ambient Pauli-vector space, no fixed orthogonal transfer can decrease $F_{\alpha}$ for all inputs; hence a fortiori no state-independent universal rule exists on the physical subset. The effectiveness of the next layer necessarily depends on the current Pauli spectrum produced by the previous layers.

Furthermore, this no-go result also explains why the first-layer preparation discussed in Sec.~\ref{sec:shallow_layer} continues to matter in the multilayer setting. Since each new SQR block acts only by redistributing and mixing the existing Pauli coefficients, later layers inherit the spectral structure created earlier rather than replacing it. Therefore, optimizing magic generation in early FTQC must be done adaptively and layer by layer. This observation motivates the next subsection, where we formulate an explicit state-dependent optimization strategy for choosing the analog angles and Clifford mixing at each step.

\subsection{Spectrum-aware optimization of the next SQR layer}

The preceding no-go theorem eliminates the possibility of a universal, state-independent rule for selecting the next SQR block to guarantee monotonic magic enhancement. Consequently, the relevant question shifts: whether such a fixed rule exists, but rather how the $(N+1)$th layer should be chosen once the current Pauli spectrum is known. Equation~\eqref{eq:STAR_iterative_update} provides the answer: the next SQR block must be selected adaptively from the current Pauli-coefficient vector.

A crucial structural point is that the Clifford and analog parts of an SQR layer play fundamentally different roles. A Clifford layer does not create new amplitudes; it acts only by a signed symplectic permutation of the Pauli coefficients. By contrast, the analog $Z$-rotation layer performs the actual redistribution of spectral weight. In this sense, the optimization of the next SQR block is naturally hybrid: the Clifford layer rearranges the spectrum, while the analog layer carries out the local flattening.

More precisely, Eq.~\eqref{eq:STAR_iterative_update} shows that the analog SQR layer acts independently within each fixed-$\mathbf x$ sector. For a given $(\mathbf x,\mathbf z)$, the updated coefficient receives contributions only from terms labeled by $\mathbf u \le \mathbf x$, so that coefficients with the same $X$-support $\mathbf x$ but different $Z$ components are mixed together. Equivalently, the analog step may be viewed as a spectrum-spreading process over the family of Pauli strings that share a common $\mathbf x$. The size of the accessible mixing space is therefore determined by the number of allowed subsets $\mathbf u \le \textbf x$, namely $2^{|\mathbf x|}$. Thus, the Hamming weight $|\mathbf x|$ directly quantifies the mixing capacity of the sector, or equivalently, the Pauli support rank of the corresponding block. Sectors with larger support admit broader redistribution of spectral weight, whereas the $x=0^n$ sector remains completely unmixed.

This immediately yields a spectrum-aware design principle for the Clifford part of the next SQR block. Since the subsequent analog evolution has greater spreading power in sectors with larger Pauli support rank, the Clifford layer should be chosen to relocate those coefficients that deviate most strongly from the target flattened profile into sectors with larger $|x|$. Conversely, coefficients already close to the target distribution should be left in sectors with smaller support, where they are disturbed less. In this sense, the Clifford layer acts as a preconditioner for the restricted analog mixing.

Once this preconditioning is fixed, the analog part of the next SQR block is determined by a continuous optimization over the single-qubit rotation angles $\mathbf w$. For a given preconditioned spectrum, the optimal analog layer is obtained by minimizing the resulting $F_\alpha$:
\begin{equation}
\min_{\mathbf{w}\in[0,1)^n}
\hspace{-1ex}f(\mathbf{w})
:=
\hspace{-2.5ex}\sum_{(\mathbf{x},\mathbf{z})\in\mathbb{F}_2^{2n}}
\hspace{-0.5ex}\bigg|\hspace{-0.5ex}
\sum_{\mathbf{u}\le \mathbf{x}}
\Gamma_{\mathbf{x}}(\mathbf{u};\mathbf{w})
(-1)^{\mathbf{u}\cdot \mathbf{z}} a_{P(\mathbf{x},\mathbf{z}\oplus \mathbf{u})}
\bigg|^{2\alpha}
\label{eq:analog_layer_optimization}
\end{equation}
with $w_j\in[0,1)$ for $j=1,\dots,n$.
Here $f(\mathbf w)$ is precisely the post-layer $F_\alpha$ functional expressed through the SQR transfer coefficients $\Gamma_{\mathbf x}(\mathbf u;\mathbf w)$. Choosing the analog layer is therefore equivalent to choosing the local mixing amplitudes that most effectively flatten the current Pauli spectrum. 

A highly advantageous feature of Eq.~\eqref{eq:analog_layer_optimization} is that the objective function decouples across the fixed-$\mathbf{x}$ sectors. Although the global optimization landscape is inherently non-convex, it remains entirely analytically differentiable with respect to the continuous angles $\mathbf{w}$. This makes the analog optimization directly amenable to standard gradient-based classical solvers, such as L-BFGS or Adam~\cite{f186b552-bbfe-3ddc-b531-329e36ac4b69,kingma2017adammethodstochasticoptimization}, supplemented by multiple random restarts to avoid local minima.

When the Clifford preconditioning and analog mixing are combined, the full update takes the form
\begin{align}
a'_{P(x,z)}
=
&\sum_{\mathbf u\le\mathbf x}
\Gamma_{\mathbf x}(\mathbf u;\mathbf w)\,
(-1)^{|\mathbf u|+\mathbf u\cdot \mathbf z+h(F^{-1}(x,z\oplus u))}
\notag \\ 
&\times a_{F^{-1}P(\mathbf x,\mathbf z\oplus \mathbf u)},
\label{eq:full_star_update_with_clifford}
\end{align}
where $h(\cdot)$ accounts for the phase factors induced by Clifford conjugation. Eq.~\eqref{eq:full_star_update_with_clifford} makes the logic of the optimization transparent: the Clifford layer first redistributes Pauli weight among sectors of different support rank, and the analog layer then performs the strongest local spreading allowed within each sector.

Ultimately, this formulation turns magic optimization in the $N$-layer model (ansatz~2) into an explicit state-dependent design strategy for early FTQC circuits. Rather than relying on any universal gate-selection rule, one should choose each new SQR block adaptively from the current Pauli spectrum: first use the Clifford layer to precondition the spectrum according to sector-dependent mixing capacity, and then optimize the analog angles to achieve the strongest local flattening permitted by the transfer rule in Eq.~\eqref{eq:STAR_iterative_update}.

At the same time, the structure of Eq.~\eqref{eq:STAR_iterative_update} already reveals an intrinsic limitation of the ansatz~2. Since the analog layer mixes only within fixed-$x$ sectors, its spreading power is fundamentally constrained by its ability to enlarge the nonzero Pauli spectrum. The next subsection shows that this restriction is not merely technical but leads to a genuine architectural gap between SQR-gate choice and the unrestricted shallow diagonal mechanism of Sec.~\ref{sec:shallow_layer}.

\subsection{Architectural limitation and improvement}

We now return to the architectural question underlying this section: even with optimal layer-by-layer tuning, can the ansatz~2 generate magic as efficiently as the unrestricted shallow diagonal model of Sec.~\ref{sec:shallow_layer}? The answer is no. The reason is structural: the iterative SQR update in Eq.~\eqref{eq:STAR_iterative_update} is built entirely from restricted local mixing of the existing Pauli spectrum. As a result, the magic-generation capacity of the full $N$-layer architecture is governed by how efficiently each layer can enlarge the set of nonzero Pauli expectation values. The first non-Clifford layer is especially important because it determines the initial spectral pattern on which all subsequent restricted mixing acts.

To quantify this limitation, we introduce the Pauli support size
\begin{equation}
    \mathcal{N}{\mathrm{nz}}
    :=
    \sum_{\mathbf{x}\in\mathbb{F}_2^n}
    \left|
    \left\{
    \mathbf{z}\in\mathbb{F}_2^n \;\middle|\;
    a_{P(\mathbf{x},\mathbf{z})}\neq 0
    \right\}
    \right|,
    \label{eq:pauli_support}
\end{equation}
which counts the number of nonzero Pauli expectation values. Mathematically, this quantity evaluates the zeroth moment of the Pauli spectrum, 
$\lim_{\alpha \to 0} F_\alpha = \mathcal{N}_{\mathrm{nz}}$. This quantity measures the geometric expressibility of the architecture at the level of Pauli spreading.

For ansatz~2, the shallow-layer Pauli spectrum factorizes because the native phase polynomial is a linear phase polynomial. This locality severely restricts the number of nonzero Pauli spectra. Maximizing over the rotation parameters yields
\begin{equation}
    \mathcal{N}_{\mathrm{nz}}^{\max}(\mathrm{SQR}) = 3^n.
\end{equation}
By contrast, the unrestricted shallow diagonal model can, in principle, populate the entire non-identity Pauli sector,
\begin{equation}
    \mathcal{N}_{\mathrm{nz}}^{\max}(\mathrm{diag}) = 4^n - 2^n.
\end{equation}
Hence, for $n>1$,
\begin{equation}
\mathcal{N}_{\mathrm{nz}}^{\max}(\mathrm{SQR}) = 3^n
\;<\;
4^n - 2^n
=
\mathcal{N}_{\mathrm{nz}}^{\max}(\mathrm{diag}).
\end{equation}
So SQR falls short of the unrestricted model by an exponential factor. 

This comparison identifies a genuine architectural bottleneck. In Sec.~\ref{sec:shallow_layer}, we saw that maximal magic generation requires both broad Pauli support and sufficiently flat amplitudes. SQR already fails at the first requirement: it cannot generate the support maximality available to general diagonal-hierarchy gates. Since the $N$-layer model (ansatz~2) evolution is obtained by iterating the same restricted transfer rule, this shallow-layer support limitation remains and even amplifies the relevant bottleneck for the multilayer architecture as well.

The implication for hardware design is immediate. The limitation of using SQR (ansatz~2) is not simply that it uses fewer non-Clifford degrees of freedom, but that its native diagonal layer is structurally too linear. To overcome this bottleneck, one must enlarge the phase polynomial beyond the linear form. A minimal way to do so is to replace isolated single-qubit $Z$ rotations by nonlinear diagonal interactions, such as multi-qubit $Z$ rotations ~\cite{Litinski_2019},
or more general multi-qubit diagonal couplings. These interactions introduce the nonlinear phase structure needed to increase Pauli spreading and thereby improve the magic-generation capacity of the architecture.

Therefore, the ansatz~2 analysis leads to a concrete design principle for early FTQC: maximizing logical non-stabilizer resource generation requires going beyond the linear single-qubit diagonal layer. While spectrum-aware optimization is necessary to guarantee increasing magic monotonicity, overcoming the fundamental kinematic bottleneck for early FTQC architectures that use SQR as their non-Clifford gate. It ultimately demands hardware support for nonlinear diagonal interactions.

\section{Conclusion and outlook}
\label{sec:conclusion}
In this work, we have established a comprehensive analytical framework for the early FTQC design problem: maximizing the generation of quantum magic in logical circuits composed of alternating Clifford layers and diagonal non-Clifford resources. By utilizing Pauli-spectrum functionals to strictly quantify both quantum magic and classical simulability, we demonstrate that resource optimization cannot be reduced to simple algebraic heuristics. Neither the gate class nor the hierarchy level is sufficient to predict magic growth without explicit reference to the dynamically evolving Pauli spectrum.

In the shallow-layer setting, we identified when a diagonal hierarchy gate fails to generate magic on a stabilizer input and characterized the structure of stabilizer states that achieve the flattest attainable spectra. This analysis yields our first no-go theorem: there exists no state-independent gate-selection rule—including any metric based solely on the diagonal-hierarchy level—that guarantees the monotonic improvement of Pauli-spectrum magic measures. Extending this to the $N$-layer model (ansatz~2), we derived an explicit iterative update of Pauli coefficients, revealing that each new block acts as a restricted redistribution of the existing spectrum. This establishes our second no-go theorem: universal Pauli-spectrum flattening under fixed orthogonal transfer is mathematically impossible. Any layer choice that decreases $F_\alpha$ for a certain subset of inputs must inevitably increase it for others. Together, these no-go theorems fundamentally recast early FTQC design from a static instruction sequence into an adaptive control problem: the optimal subsequent layer must be synthesized in a spectrum-aware manner, guided directly by the current Pauli distribution.

Crucially, these structural insights also provide a constructive, state-dependent design rule. The first non-Clifford layer is decisive because it sets the accessible support on which later layers can only reshuffle and refine. At depth $(N\!+\!1)-$th, the Clifford layer should be used as a preconditioner that permutes the largest-magnitude Pauli coefficients into large-$|\textbf x|$ sectors, where the subsequent mixing is strongest, while leaving near-target coefficients in small-$|\textbf x|$ sectors. Conditioned on this permutation, the analog single-qubit $Z$ angles $\mathbf{w}$ can be selected by continuous optimization of the differentiable objective induced by the transfer rule of ansatz~2. This hybrid discrete--continuous strategy provides a practical, spectrum-aware workflow for designing partially fault-tolerant circuits.

Finally, our analysis exposes a severe architecture limitation within the class of architectures represented by ansatz~2: the linear phase structure of parallel single-qubit $Z$ rotations imposes a maximal Pauli-support ceiling $\mathcal{N}_{\mathrm{nz}}^{\max}(\mathrm{SQR}) = 3^n$, which is exponentially smaller than the $(4^n-2^n)$ non-identity support available to general diagonal hierarchy gates. Therefore, even perfect spectrum-aware optimization cannot close the expressibility gap unless the native diagonal layer is upgraded beyond a linear-phase polynomial. A minimal upgrade is to include nonlinear diagonal interactions such as multi-qubit $Z$ rotations, which break the linear bottleneck and provide a clear hardware-level route toward stronger magic generation in early FTQC architectures. Our results thus establish magic capacity—complementing physical error suppression—as the foundational benchmark for evaluating an architecture’s logical reach and its potential for quantum advantage.

Looking forward, our results offer concrete prescriptions for both experimental implementations and architectural benchmarking. We recommend evaluating each logical layer by: (i) estimating a small set of Pauli moments (or a classical-shadow summary) sufficient to track $F_\alpha$ and the nonzero Pauli support size, (ii) applying a greedy or randomized Clifford preconditioner targeting large-$|x|$ sectors, and (iii) optimizing $\mathbf{w}$ via gradient-based search with restarts under hardware constraints. Comparing single-qubit rotations against multi-qubit $Z$ rotations at a fixed depth provides a direct experimental test of the predicted support bottleneck and the benefits of introducing nonlinear diagonal phases.

\section*{Acknowledgments}
We thank Bartosz Reguła, and Ryuji Takagi for useful discussions. HHL acknowledges financial support from the Japan Science and Technology Agency (JST) SPRING program (Grant No. JPMJSP2108). YS is supported by MEXT Q-LEAP Grant No.~JPMXS0120319794, MEXT Feasibility Study on the future HPCI, JST Moonshot R\&D Grant No.~JPMJMS2061, JST CREST Grant No.~JPMJCR23I4, JPMJCR24I4, and JPMJCR25I4. EJK acknowledges financial support from the National Science and Technology Council (NSTC) of Taiwan under Grant No.~NSTC~114-2112-M-A49-036-MY3. This research was partly funded by the Ministry of Education, Culture, Sports, Science and Technology (MEXT) Quantum Leap Flagship Program (Q-LEAP) (Grant No.~JPMXS0118068682) and the Japan Science and Technology Agency (JST) as part of Adopting Sustainable Partnerships for Innovation Research Ecosystem (ASPIRE) (Grant No.~JPMJAP2513).

\section*{Author Contributions}

HHL conceived the project, identified the primary research direction, proposed the central concepts underlying Main Results I--III, and derived the analytical solutions and proofs for Main Results II and III. YS formalized the optimization problem and established the theoretical framework utilized throughout the manuscript. EJK completed the rigorous proof for Main Result I and developed the optimization strategy for Main Result III. YN supervised the project. All authors discussed the results and contributed to the writing and revision of the manuscript.

\bibliography{references}

\appendix
\widetext

\section*{SUPPLEMENTAL MATERIAL}

\section{An \texorpdfstring{$L^p$}{Lp} hierarchy of stabilizer-expansion quantifiers}
\label{app:Lp_stabilizer_family}

Several standard magic quantifiers, including the robustness of magic (RoM)~\cite{hamaguchi2024handbook}, stabilizer extent~\cite{heimendahl2021stabilizer}, and stabilizer rank~\cite{peleg2022lower}, arise naturally from decompositions of a pure state into stabilizer states,
\begin{equation}
|\psi\rangle=\sum_i c_i |s_i\rangle,
\qquad
|s_i\rangle\in\mathrm{STAB}_n,
\end{equation}
where $\mathrm{STAB}_n$ denotes the set of $n$-qubit stabilizer states and $c_i\in\mathbb C$. These quantities may be organized into a unified $L^p$-type family for pure-state magic~\cite{Regula_2017}. For $p>0$, define
\begin{equation}
\gamma_p(|\psi\rangle)
:=
\inf\Bigl\{
\|c\|_p:
|\psi\rangle=\sum_i c_i |s_i\rangle
\Bigr\},
\end{equation}
where
\begin{equation}
\|c\|_p
=
\left(\sum_i |c_i|^p\right)^{1/p}.
\end{equation}

Important special cases are obtained in familiar limits. For $p=1$, $\gamma_1$ is the stabilizer gauge, whose square coincides with the stabilizer extent,
\begin{equation}
\xi(|\psi\rangle)=\gamma_1(|\psi\rangle)^2,
\end{equation}
and, for pure states, also coincides with the robustness of magic. In the formal \(p\to0\) limit, the objective \(\sum_i |c_i|^p\) approaches the support size of the decomposition. In this sense, the \(L^p\) family connects to the stabilizer-rank optimization problem, which asks for the minimum number of stabilizer states appearing in a decomposition corresponding to the minimal support size of a stabilizer decomposition. Unlike entropic magic measures, the family $\gamma_p$ is defined through atomic decompositions and therefore has a more geometric character. In general, these quantities are submultiplicative under tensor products, but not multiplicative.

\begin{theorem}[$L^1$ upper bound for subadditive stabilizer functionals]
\label{thm:maximality_gamma1}
Let $M(|\psi\rangle)$ be a function on state vectors satisfying:
\begin{enumerate}
    \item \textbf{Stabilizer normalization:}
    \begin{equation}
    M(|s\rangle)=1
    \qquad
    \text{for all } |s\rangle\in\mathrm{STAB}_n.
    \end{equation}

    \item \textbf{Positive homogeneity:}
    \begin{equation}
    M(\lambda|\psi\rangle)=|\lambda|\,M(|\psi\rangle)
    \qquad
    \text{for all } \lambda\in\mathbb{C}.
    \end{equation}

    \item \textbf{Subadditivity:}
    \begin{equation}
    M(|\phi\rangle+|\psi\rangle)
    \le
    M(|\phi\rangle)+M(|\psi\rangle).
    \end{equation}
\end{enumerate}
Then, for every pure state $|\psi\rangle$,
\begin{equation}
M(|\psi\rangle)\le \gamma_1(|\psi\rangle).
\end{equation}
In this sense, $\gamma_1$ is maximal among all stabilizer-normalized seminorm-like functionals on state vectors.
\end{theorem}

\begin{proof}
Let
\begin{equation}
|\psi\rangle=\sum_i c_i |s_i\rangle
\end{equation}
be an arbitrary stabilizer decomposition. By repeated subadditivity,
\begin{equation}
M(|\psi\rangle)
=
M\!\left(\sum_i c_i |s_i\rangle\right)
\le
\sum_i M(c_i |s_i\rangle).
\end{equation}
Using positive homogeneity and stabilizer normalization,
\begin{equation}
\sum_i M(c_i |s_i\rangle)
=
\sum_i |c_i|\,M(|s_i\rangle)
=
\sum_i |c_i|.
\end{equation}
Hence,
\begin{equation}
M(|\psi\rangle)\le \sum_i |c_i|.
\end{equation}
Since this holds for every stabilizer decomposition of $|\psi\rangle$, taking the infimum over all such decompositions gives
\begin{equation}
M(|\psi\rangle)\le \gamma_1(|\psi\rangle).
\end{equation}
\end{proof}

\begin{theorem}[Upper bound by the stabilizer $L^p$ gauge]
\label{thm:maximality_gamma_p}
Let $p>0$. Let $M_p(|\psi\rangle)$ be a function on state vectors satisfying:
\begin{enumerate}
    \item \textbf{Stabilizer normalization:}
    \begin{equation}
    M_p(|s\rangle)=1
    \qquad
    \text{for all } |s\rangle\in\mathrm{STAB}_n.
    \end{equation}

    \item \textbf{Positive homogeneity:}
    \begin{equation}
    M_p(\lambda|\psi\rangle)=|\lambda|\,M_p(|\psi\rangle)
    \qquad
    \text{for all } \lambda\in\mathbb{C}.
    \end{equation}

    \item \textbf{$p$-subadditivity:}
    \begin{equation}
    M_p(|\phi\rangle+|\psi\rangle)^p
    \le
    M_p(|\phi\rangle)^p+M_p(|\psi\rangle)^p.
    \end{equation}
\end{enumerate}
Then, for every pure state $|\psi\rangle$,
\begin{equation}
M_p(|\psi\rangle)\le \gamma_p(|\psi\rangle).
\end{equation}
In this sense, $\gamma_p$ is maximal among all stabilizer-normalized
$p$-subadditive homogeneous functionals on state vectors.
\end{theorem}

\begin{proof}
Let
\begin{equation}
|\psi\rangle=\sum_i c_i |s_i\rangle
\end{equation}
be an arbitrary stabilizer decomposition. By repeated $p$-subadditivity,
\begin{equation}
M_p(|\psi\rangle)^p
=
M_p\!\left(\sum_i c_i |s_i\rangle\right)^p
\le
\sum_i M_p(c_i |s_i\rangle)^p.
\end{equation}
Using positive homogeneity and stabilizer normalization,
\begin{equation}
\sum_i M_p(c_i |s_i\rangle)^p
=
\sum_i |c_i|^p M_p(|s_i\rangle)^p
=
\sum_i |c_i|^p.
\end{equation}
Hence,
\begin{equation}
M_p(|\psi\rangle)^p \le \sum_i |c_i|^p.
\end{equation}
Since this holds for every stabilizer decomposition of $|\psi\rangle$, taking the infimum over all such decompositions gives
\begin{equation}
M_p(|\psi\rangle)^p
\le
\inf_{\psi=\sum_i c_i s_i}\sum_i |c_i|^p
=
\gamma_p(|\psi\rangle)^p.
\end{equation}
Therefore,
\begin{equation}
M_p(|\psi\rangle)\le \gamma_p(|\psi\rangle).
\end{equation}
\end{proof}

\begin{remark}[$L^p$ maximality and natural appearance]
The preceding theorem shows that each $\gamma_p$ can be viewed as the maximal stabilizer-normalized, positively homogeneous functional compatible with a $p$-type decomposition rule. The case $p=1$ reduces to ordinary convexity and coincides with the stabilizer gauge/extent, while other $p\neq 1$ impose an $L^p$-geometry on the decomposition. Clifford-invariance rules out assigning different intrinsic weights to individual stabilizer states within the Clifford orbit. This makes symmetric functions of the coefficient magnitudes, and in particular the \(L^p\) costs, natural choices for stabilizer-expansion quantifiers.

The preceding result does not imply that a natural stabilizer-expansion quantifier must be a single \(L^p\) cost.  One may also consider symmetric combinations of several such costs, for example
\[
M(|\psi\rangle)
=
\max\{\gamma_{p_1}(|\psi\rangle),\gamma_{p_2}(|\psi\rangle)\},
\qquad
M(|\psi\rangle)
=
\sum_a w_a \gamma_{p_a}(|\psi\rangle),
\]
or more generally, functions of \((\gamma_{p_1},\gamma_{p_2},\ldots)\) that preserve the desired homogeneity and monotonicity properties. Thus, the \(L^p\) family should be viewed as a natural hierarchy of symmetric stabilizer-decomposition costs, rather than as a uniqueness classification of all possible magic monotones.
\end{remark}

We next record the basic tensor-product property of the $L^p$ family.

\begin{proposition}[Tensor submultiplicativity of the $L^p$ stabilizer gauges]
\label{prop:gamma_p_submult}
For any $p>0$ and any pure states $|\psi\rangle,|\phi\rangle$, the stabilizer $L^p$ gauge satisfies
\begin{equation}
\gamma_p(|\psi\rangle\otimes|\phi\rangle)
\le
\gamma_p(|\psi\rangle)\,\gamma_p(|\phi\rangle).
\end{equation}
\end{proposition}

\begin{proof}
Fix $\varepsilon>0$. Choose stabilizer decompositions
\begin{equation}
|\psi\rangle=\sum_i c_i |s_i\rangle,
\qquad
|\phi\rangle=\sum_j d_j |t_j\rangle,
\end{equation}
such that
\begin{equation}
\|c\|_p\le \gamma_p(|\psi\rangle)+\varepsilon,
\qquad
\|d\|_p\le \gamma_p(|\phi\rangle)+\varepsilon.
\end{equation}
Taking the tensor product yields
\begin{equation}
|\psi\rangle\otimes|\phi\rangle
=
\sum_{i,j} c_i d_j \,
\bigl(|s_i\rangle\otimes |t_j\rangle\bigr),
\end{equation}
and each tensor product $|s_i\rangle\otimes|t_j\rangle$ is again a stabilizer state.

Let $e$ denote the coefficient array $e_{ij}=c_i d_j$. Then
\begin{equation}
\|e\|_p^p
=
\sum_{i,j}|c_i d_j|^p
=
\left(\sum_i |c_i|^p\right)
\left(\sum_j |d_j|^p\right)
=
\|c\|_p^p\,\|d\|_p^p,
\end{equation}
so
\begin{equation}
\|e\|_p=\|c\|_p\,\|d\|_p.
\end{equation}
By the definition of $\gamma_p$,
\begin{equation}
\gamma_p(|\psi\rangle\otimes|\phi\rangle)
\le
\|e\|_p
=
\|c\|_p\,\|d\|_p
\le
\bigl(\gamma_p(|\psi\rangle)+\varepsilon\bigr)
\bigl(\gamma_p(|\phi\rangle)+\varepsilon\bigr).
\end{equation}
Taking $\varepsilon\to 0$ proves the claim.
\end{proof}

This tensor-product behavior also distinguishes the individual members of the \(L^p\) hierarchy. Although one may form symmetric or general linear combinations of several $\gamma_p$, such combinations do not, in general, preserve the strict tensor submultiplicativity exhibited by each individual $\gamma_p$. Thus, the single-parameter family \(\{\gamma_p\}_{p>0}\) provides the natural basic building blocks for stabilizer-expansion costs.

\section{Classification of Pauli-spectrum magic functionals}
\label{app:pauli_spectrum_classification}

In this appendix, we classify continuous pure-state magic functionals that depend only on the Pauli expectation values and are additive under tensor products. The main result is that this class is exhausted by logarithmic combinations of the moments
\begin{equation}
F_\alpha(|\psi\rangle)
:=
\sum_{P\in\mathcal P_n}
|\langle\psi|P|\psi\rangle|^{2\alpha},
\end{equation}
or equivalently by linear combinations of stabilizer Rényi entropies.

We begin with the standard fact that continuous functions on the compact domain of Pauli expectation values may be uniformly approximated by polynomials.

\begin{theorem}[Stone--Weierstrass theorem~\cite{stone1937applications}]
Let $K\subset \mathbb R^m$ be compact. Then every continuous real-valued function on $K$ can be uniformly approximated by polynomials.
\end{theorem}

In our setting, the Pauli expectation vector
\begin{equation}
a_P := \langle \psi|P|\psi\rangle
\end{equation}
ranges over a bounded subset of $\mathbb R^{4^n}$, since $|a_P|\le 1$ for all $P\in\mathcal P_n$. Hence, any continuous Pauli-spectrum functional can be approximated uniformly by polynomial functions of the variables $\{a_P\}$.

We also use the standard characterization of symmetric polynomials.

\begin{lemma}[Symmetric polynomials generated by power sums]
\label{lemma:symmetric_power_sum_app}
Let $x_1,\dots,x_m$ be variables. Any symmetric polynomial in $x_1,\dots,x_m$ can be expressed as a polynomial in the power sums
\begin{equation}
p_k=\sum_{i=1}^m x_i^k,
\qquad k=1,\dots,m.
\end{equation}
\end{lemma}

We can now prove the main classification result.

\begin{theorem}[Classification of symmetric tensor-additive Pauli-spectrum functionals]
\label{thm:classification_app}
Let $Q$ be a continuous function of the Pauli expectation values
\begin{equation}
x_P := |\langle\psi|P|\psi\rangle|^2,
\end{equation}
satisfying:
\begin{enumerate}
    \item \textbf{Symmetry:} $Q$ is symmetric under permutations of the variables $x_P$;
    \item \textbf{Tensor additivity:}
    \begin{equation}
        Q(\psi\otimes\phi)=Q(\psi)+Q(\phi).
    \end{equation}
\end{enumerate}
Then $\forall \epsilon>0$ there exist real coefficients $C_\alpha$ with finite support such that
\begin{equation}
    \left| Q(\psi)-
    \sum_{\alpha\ge 0} C_\alpha \log F_\alpha(\psi)\right| \leq \epsilon
\end{equation}
where
\begin{equation}
    F_\alpha(\psi)
    :=
    \sum_{P\in\mathcal P_n}
    |\langle\psi|P|\psi\rangle|^{2\alpha}.
\end{equation}
The term $\alpha=0$ contributes only the dimension-dependent constant
\begin{equation}
F_0 = |\mathcal P_n| = d^2.
\end{equation}
\end{theorem}

\begin{proof}
Define
\begin{equation}
F_{\mathrm{eff}}(\psi):=e^{Q(\psi)}.
\end{equation}
Because $Q$ is continuous on the bounded domain of Pauli expectation values, $F_{\mathrm{eff}}$ is also continuous. By the Stone--Weierstrass theorem, $F_{\mathrm{eff}}$ can be uniformly approximated by polynomials in the variables $\{x_P\}$. Since $Q$ is symmetric in the $x_P$, the same is true for $F_{\mathrm{eff}}$, and by Lemma~\ref{lemma:symmetric_power_sum_app} any such symmetric polynomial can be expressed as a polynomial in the power sums
\begin{equation}
F_\alpha(\psi)=\sum_{P\in\mathcal P_n} x_P^\alpha.
\end{equation}

Tensor additivity of $Q$ implies multiplicativity of $F_{\mathrm{eff}}$:
\begin{equation}
F_{\mathrm{eff}}(\psi\otimes\phi)
=
F_{\mathrm{eff}}(\psi)\,F_{\mathrm{eff}}(\phi).
\end{equation}
At the same time, each $F_\alpha$ is multiplicative under tensor products,
\begin{equation}
F_\alpha(\psi\otimes\phi)=F_\alpha(\psi)\,F_\alpha(\phi),
\end{equation}
since Pauli expectation values factorize on product states. Therefore, any monomial in the variables $\{F_\alpha\}$ is multiplicative. Conversely, a nontrivial sum of distinct monomials cannot satisfy multiplicativity identically for arbitrary product states. It follows that the multiplicative symmetric polynomial $F_{\mathrm{eff}}$ must take the form
\begin{equation}
F_{\mathrm{eff}}(\psi)
=
C\prod_{\alpha\ge 1}F_\alpha(\psi)^{k_\alpha}
\end{equation}
for some real exponents $k_\alpha$ and constant $C>0$.

Now note that
\begin{equation}
F_1(\psi)=\sum_{P\in\mathcal P_n}|\langle\psi|P|\psi\rangle|^2=d,
\end{equation}
which is state-independent for pure states. Hence, the $\alpha=1$ contribution can be absorbed into a dimension-dependent prefactor, which we write as
\begin{equation}
C=d^q.
\end{equation}
Taking logarithms yields
\begin{equation}
Q(\psi)
=
q\log d-\sum_{\alpha\ge 2}k_\alpha \log F_\alpha(\psi),
\end{equation}
which proves the claim.
\end{proof}

It is often convenient to rewrite the result in terms of stabilizer Rényi entropies.

\begin{definition}[Stabilizer $\alpha$-Rényi entropy]
\label{def:SRE_app}
For an $n$-qubit pure state $|\psi\rangle$, define
\begin{equation}
M_\alpha(|\psi\rangle)
=
\frac{1}{1-\alpha}
\log\!\left(
d^{-\alpha}
\sum_{P\in\mathcal P_n}
|\langle\psi|P|\psi\rangle|^{2\alpha}
\right)
-\log d,
\end{equation}
where $d=2^n$.
\end{definition}

Since
\begin{equation}
\log F_\alpha
=
(1-\alpha)M_\alpha+\alpha\log d,
\end{equation}
Theorem~\ref{thm:classification_app} is equivalent to the statement that every continuous symmetric tensor-additive Pauli-spectrum functional can be written as a finite linear combination of stabilizer Rényi entropies, up to a dimension-dependent term.

As an immediate consequence, stabilizer nullity lies in the same class.

\begin{corollary}[Stabilizer nullity as a Rényi-type functional]
\label{cor:nullity_in_class_app}
The stabilizer nullity
\begin{equation}
\nu(|\psi\rangle)
:=
n-\log_2 s(|\psi\rangle),
\end{equation}
where
\begin{equation}
s(|\psi\rangle)
=
\bigl|
\{P\in\mathcal P_n:\,P|\psi\rangle=\pm|\psi\rangle\}
\bigr|,
\end{equation}
is obtained as the limit
\begin{equation}
\nu(|\psi\rangle)=\lim_{\alpha\to\infty}M_\alpha(|\psi\rangle).
\end{equation}
\end{corollary}

\begin{proof}
By definition,
\begin{equation}
F_\alpha(|\psi\rangle)=\sum_{P\in\mathcal P_n}|\langle\psi|P|\psi\rangle|^{2\alpha}.
\end{equation}
As $\alpha\to\infty$, only those terms with $|\langle\psi|P|\psi\rangle|=1$ survive, so
\begin{equation}
\lim_{\alpha\to\infty}F_\alpha(|\psi\rangle)=s(|\psi\rangle).
\end{equation}
Substituting into the definition of $M_\alpha$ gives
\begin{equation}
\lim_{\alpha\to\infty}M_\alpha(|\psi\rangle)
=
n-\log_2 s(|\psi\rangle)
=
\nu(|\psi\rangle),
\end{equation}
as claimed.
\end{proof}

\begin{corollary}[Mana in the closure of the Pauli-spectrum family]
For odd-prime qudit systems, mana is contained in the uniform closure of the Pauli-spectrum functional family.
\end{corollary}

\begin{proof}
By definition, each discrete Wigner coefficient can be expressed via the symplectic Fourier transform as
\begin{equation}
W_\rho(\mathbf b)=\sum_{\mathbf a} c_{\mathbf b,\mathbf a}\,\Tr(\rho T_{\mathbf a}),
\end{equation}
where $T_{\mathbf{a}}$ are the Heisenberg--Weyl operators~\cite{Ahmadi_2024}. Thus, each $W_\rho(\mathbf b)$ is strictly a linear combination of Pauli-spectrum expectation values. 

Because the space of valid physical density matrices is compact, the physical domain of these expectation values is a compact subset of $\mathbb{R}^{d^2}$. Since the absolute value map \(x\mapsto |x|\) is continuous on this compact domain, the Stone-Weierstrass theorem guarantees that each \(|W_\rho(\mathbf b)|\) can be uniformly approximated by polynomials in the Pauli spectrum. 

Consequently, the sum $\sum_{\mathbf b}|W_\rho(\mathbf b)|$ is also uniformly approximated by such polynomials. Furthermore, the normalization of the Wigner function ($\sum_{\mathbf{b}} W_\rho(\mathbf{b}) = 1$) ensures that $\sum_{\mathbf b}|W_\rho(\mathbf b)| \geq 1$. Because the logarithm is uniformly continuous on the interval $[1, \infty)$, the composition
\begin{equation}
\mathcal M(\rho)=\log \left( \sum_{\mathbf b}|W_\rho(\mathbf b)| \right)
\end{equation}
preserves this uniform convergence. Therefore, mana is a continuous functional strictly contained within the uniform closure of the polynomial Pauli-spectrum functionals.
\end{proof}

\section{The Shallow Layer Model (Ansatz 1)}
\label{app:shallow_analytic}

We begin by expressing arbitrary stabilizer states using the symplectic formalism. For any binary vectors $(\mathbf{x}, \mathbf{z}) \in \mathbb{F}_2^{2n}$, the $n$-qubit Pauli operator is defined as
\begin{equation}
    P(\mathbf{x}, \mathbf{z}) := i^{\mathbf{x}\cdot \mathbf{z}} X^{\mathbf{x}} Z^{\mathbf{z}},
\end{equation}
where $X^{\mathbf{x}} := X_1^{x_1} \cdots X_n^{x_n}$ and $Z^{\mathbf{z}} := Z_1^{z_1} \cdots Z_n^{z_n}$. The commutation relation between two Pauli operators is dictated by the symplectic inner product $\omega((\mathbf{x}, \mathbf{z}), (\mathbf{x}', \mathbf{z}')) = \mathbf{x}\cdot \mathbf{z'} + \mathbf{z}\cdot \mathbf{x'} \pmod{2}$. Specifically, their product is
\begin{equation}
    P(\mathbf{x}, \mathbf{z}) P(\mathbf{x}', \mathbf{z}')
= i^{\mathbf{z}\cdot \mathbf{x}' - \mathbf{z}'\cdot \mathbf{x}}P(\mathbf{x}+\mathbf{x}', \mathbf{z}+ \mathbf{z}').
\end{equation}

A pure stabilizer state $|\psi\rangle$ is uniquely specified by an $n$-dimensional isotropic (Lagrangian) subspace $\mathbb{L}\subset \mathbb{F}_2^{2n}$, satisfying
\begin{equation}
\omega((\mathbf x, \mathbf z),(\mathbf x', \mathbf z'))=0,
\qquad
\forall\,(\mathbf x, \mathbf z),(\mathbf x', \mathbf z')\in\mathbb{L}.
\end{equation}
Equivalently, choosing a stabilizer generator matrix $G$ whose rows are the symplectic generator labels $(u_i,v_i)$, the projector onto $|\psi\rangle$ may be written as
\begin{equation}
|\psi\rangle\langle\psi|
=
\frac{1}{2^n}
\sum_{a\in\mathbb Z_2^n}
(-1)^{\sum_j h_j a_j+\sum_{k>j}(v_j\!\cdot u_k)a_j a_k}
\,P\!\left(\sum_i a_i (u_i,v_i)\right),
\end{equation}
where the binary phase vector $h\in\mathbb Z_2^n$ records the $\pm1$ signs of the corresponding stabilizer generators, following the notation introduced in the main text.

We evaluate the shallow-layer model by applying a diagonal hierarchy gate $U$ to this state. Since $U$ is diagonal, $U|\mathbf{b}\rangle = \exp(i\theta(\mathbf{b}))|\mathbf{b}\rangle$, where $\theta(\mathbf{b}) = 2\pi\sum_m 2^{-m} f_m(\mathbf{b})$. The Pauli expectation value $a_{P(\mathbf{x},\mathbf{z})} = \langle\psi| U^\dagger P(\mathbf{x},\mathbf{z}) U |\psi\rangle$ evaluates to
\begin{align}
    a_{P(\mathbf{x},\mathbf{z})} &= \sum_{\mathbf{b},\mathbf{b'}\in \mathbb{Z}_2^n} e^{i(\theta(\mathbf{b})-\theta(\mathbf{b'}))} \langle\psi|\mathbf{b}\rangle\langle\mathbf{b'}|P(\mathbf{x},\mathbf{z})|\mathbf{b}\rangle\langle\mathbf{b}|\psi\rangle \notag \\
    &= \sum_{\mathbf{b}\in \mathbb{Z}_2^n} e^{i(\theta(\mathbf{b})-\theta(\mathbf{b\oplus x}))} i^{\mathbf{x}\cdot \mathbf{z}} (-1)^{\mathbf{z}\cdot \mathbf{b}} \langle\mathbf{b\oplus x}|\psi\rangle\langle\psi|\mathbf{b}\rangle. \label{eq:shallow_pauli_exp}
\end{align}

To calculate the matrix element $\langle\mathbf{b\oplus x}|\psi\rangle\langle\psi|\mathbf{b}\rangle$, we substitute the symplectic expansion of the density matrix:
\begin{equation}
    \langle\mathbf{b\oplus x}|\psi\rangle\langle\psi|\mathbf{b}\rangle 
= \frac{1}{2^n} \sum_{a\in \mathbb{Z}_2^n}(-1)^{\sum\limits_j h_j\cdot a_j + \sum\limits_{k>j} (v_j\cdot u_k)a_j a_k } \langle b\oplus x|i^{(\sum_i u_i a_i)\cdot(\sum_i v_i a_i)}X^{\sum_i u_i a_i} Z^{\sum_i v_i a_i}|b\rangle
\end{equation}
Because $\langle \mathbf{b\oplus x}|X^{\mathbf{u}} Z^{\mathbf{v}}|\mathbf{b}\rangle = (-1)^{\mathbf{b}\cdot \mathbf{v}} \delta_{\mathbf{x},\mathbf{u}}$, the sum vanishes entirely unless $\mathbf{x}$ lies within the $X$-projection of the Lagrangian subspace, denoted $\mathbb{L}_x$. 

If $\mathbf{x} \notin \mathbb{L}_x$, then $a_{P(\mathbf{x},\mathbf{z})} = 0$. If $\mathbf{x} \in \mathbb{L}_x$, there exists a representative element $(\mathbf{x},\mathbf{z}_{\text{ref}}) \in \mathbb{L}$. The set of all elements in $\mathbb{L}$ with $X$-component $\mathbf{x}$ can be parameterized as $(\mathbf{x},\mathbf{z}_{\text{ref}}) \oplus \mathcal{Z}$, where $\mathcal{Z}$ is the subgroup of $r$ pure $Z$-type stabilizers generated by $\{(0, \mathbf{z}_i)\}_{i=1}^r$. By expanding the sum over the Boolean coefficients $\mathbf{c} \in \mathbb{F}_2^r$ of these $Z$-generators, the overlap becomes
\begin{equation}
    \langle\mathbf{b\oplus x}|\psi\rangle\langle\psi|\mathbf{b}\rangle = \frac{1}{2^n} \left[ (-1)^{s_0+ \mathbf x\cdot \mathbf z_{ref} +\mathbf b\cdot \mathbf z_{ref}}
    i^{\mathbf{x}\cdot \mathbf{z}_{\text{ref}}} \right] \sum_{\mathbf{c} \in \mathbb{F}_2^r} (-1)^{\sum\limits_{i=1}^r c_i A_i},
\end{equation}
where $A_i = h_i \oplus (\mathbf{b}\cdot \mathbf{z}_i) \oplus k_i$, and $k_i \in \{0,1\}$ corrects for the symplectic phase $\mathbf{x}\cdot \mathbf{z}_i = 2k_i \pmod 4$.

Exploiting the orthogonality of binary group characters, the summation over $\mathbf{c}$ cleanly factors:
\begin{equation}
    \sum_{\mathbf{c}\in\mathbb{F}_2^r} (-1)^{\sum\limits_{i=1}^r c_i A_i} = \prod_{i=1}^r \left(1+(-1)^{A_i}\right) = \begin{cases} 2^r & \text{if } A_i = 0 \quad \forall i \in \{1,\dots,r\}, \\ 0 & \text{otherwise}. \end{cases}
\end{equation}
Consequently, the amplitude enforces a strict affine condition on $\mathbf{b}$. Substituting this back into Eq.~\eqref{eq:shallow_pauli_exp}, we arrive at the exact analytical expression for the Pauli spectrum:
\begin{equation}
    a_{P(\mathbf x,\mathbf z)}
=
\sum_{b \in \mathbb{Z}_2^n}
e^{i(\theta(b)-\theta(b\oplus x))}
(-1)^{z\cdot b}
\frac{2^{r}}{2^n}\biggl[(-1)^{s_0+ \mathbf x\cdot \mathbf z_{ref} +\mathbf b\cdot \mathbf z_{ref}}i^{\mathbf x\cdot \mathbf z_{ref}}\bigg] \prod_i^r \delta(0,s_i \oplus (\mathbf b\cdot \mathbf z_i) \oplus k_i)
\label{eq: shallow_model_pauli_spectrum}
\end{equation}
When evaluating $|a_{P(\mathbf{x},\mathbf{z})}|^2$, the gauge-dependent global phases inherent to $\mathbf{z}_{\text{ref}}$ strictly cancel out, confirming that the magnitude of the spectrum relies solely on $\mathbf{x}$ and the underlying stabilizer group structure.

\section{Consistency Checks for the Pauli Spectrum}
\label{app:consistency_checks}

To verify the integrity of our generalized analytical derivations, we benchmark the spectrum formula against several fundamental constraints and established limiting cases.

\subsection{Normalization Condition}
We first verify that the calculated quantities strictly satisfy the Pauli spectrum normalization condition, $\sum_{(\mathbf{x},\mathbf{z})\in\mathbb{F}_2^{2n}} |a_{P(\mathbf{x},\mathbf{z})}|^2 = 2^n$, across the entire parameter space. Starting from our derived amplitude:
\begin{align}
\sum_{(\mathbf{x},\mathbf{z})\in\mathbb{F}_2^{2n}}|a_{P(\mathbf{x},\mathbf{z})}|^2
&= \sum_{\mathbf{x}\in\mathbb{L}_x}\sum_{\mathbf{b},\mathbf{b}' \in \mathbb{Z}_2^n}
e^{i(\theta(\mathbf{b})-\theta(\mathbf{b}\oplus \mathbf{x}) - \theta(\mathbf{b}') + \theta(\mathbf{b}'\oplus \mathbf{x}))} \notag \\
&\quad \times \frac{2^{2r}}{2^{2n}} \prod_{i=1}^r \delta_{0,\, s_i \oplus (\mathbf{b}\cdot \mathbf{z}_i) \oplus k_i} \prod_{j=1}^r \delta_{0,\, s_j \oplus (\mathbf{b}'\cdot \mathbf{z}_j) \oplus k_j} \sum_{\mathbf{z}\in\mathbb{Z}_2^n} (-1)^{\mathbf{z}\cdot (\mathbf{b}+\mathbf{b}')}.
\end{align}
The summation over $\mathbf{z}$ evaluates to $2^n \delta_{\mathbf{b},\mathbf{b}'}$. Enforcing $\mathbf{b} = \mathbf{b}'$ strictly cancels the complex phase terms (i.e., $e^0 = 1$), simplifying the expression to:
\begin{align}
\sum_{(\mathbf{x},\mathbf{z})\in\mathbb{F}_2^{2n}}|a_{P(\mathbf{x},\mathbf{z})}|^2
&= \sum_{\mathbf{x}\in\mathbb{L}_x} \frac{2^{2r}}{2^{n}} \sum_{\mathbf{b} \in \mathbb{Z}_2^{n}} \prod_{i=1}^r \delta_{0,\, s_i \oplus (\mathbf{b}\cdot \mathbf{z}_i) \oplus k_i}.
\end{align}
The product of Kronecker delta functions evaluates to 1 if and only if $\mathbf{b}\cdot \mathbf{z}_i = s_i \oplus k_i$ for all $i \in \{1, \dots, r\}$. Because the $r$ pure-$Z$ generators $\mathbf{z}_i$ are linearly independent, this imposes exactly $r$ linearly independent constraints on the vector $\mathbf{b}$. By the rank-nullity theorem, the solution space for $\mathbf{b}$ forms an affine subspace of dimension $n-r$, yielding exactly $2^{n-r}$ non-zero terms in the inner sum. Consequently:
\begin{align}
\sum_{(\mathbf{x},\mathbf{z})\in\mathbb{F}_2^{2n}}|a_{P(\mathbf{x},\mathbf{z})}|^2
= \sum_{\mathbf{x}\in\mathbb{L}_x} \frac{2^{2r}}{2^{n}} 2^{n-r} = \sum_{\mathbf{x}\in\mathbb{L}_x} 2^{r}.
\end{align}
Finally, the cardinality of the projected $X$-subspace $\mathbb{L}_x$ is intrinsically tied to the number of pure-$Z$ generators; specifically, $|\mathbb{L}_x| = 2^{n-r}$. Thus, the final summation evaluates to:
\begin{align}
\sum_{(\mathbf{x},\mathbf{z})\in\mathbb{F}_2^{2n}}|a_{P(\mathbf{x},\mathbf{z})}|^2 = 2^{n-r} \cdot 2^{r} = 2^n,
\end{align}
which perfectly preserves the global trace and normalizes the physical probability distribution.

\subsection{Consistency with Pure-\texorpdfstring{$Z$}{Z} Operators}
Since a diagonal gate $U$ commutes with pure-$Z$ Pauli strings $P(\mathbf{0},\mathbf{z})$, their expectation values should trivially match those of the unperturbed stabilizer state $|\psi\rangle$. Setting $\mathbf{x}=\mathbf{0}$ in our general formula naturally eliminates the diagonal phase polynomial:
\begin{align}
    a_{P(\mathbf{0},\mathbf{z})}
&= \sum_{\mathbf{b} \in \mathbb{Z}_2^n} \frac{2^{r}}{2^n} (-1)^{\mathbf{b}\cdot \mathbf{z}} \prod_{i=1}^r \delta_{0,\, s_i \oplus (\mathbf{b}\cdot \mathbf{z}_i)}.
\end{align}
This exactly matches the standard calculation for a pure stabilizer state:
\begin{equation}
    a_{P(\mathbf{0},\mathbf{z})} = \langle\psi| Z^{\mathbf{z}} |\psi\rangle = \sum_{\mathbf{b} \in \mathbb{Z}_2^n} (-1)^{\mathbf{z}\cdot \mathbf{b}} |\langle \mathbf{b}|\psi\rangle|^2.
\end{equation}
It is a well-established result~\cite{Aaronson_2004,garcía2017geometrystabilizerstates} that the probability distribution of a stabilizer state in the computational basis is uniform over its support subspace:
\begin{equation}
    |\langle \mathbf{b}|\psi\rangle|^2 = \begin{cases} 2^{r-n} & \text{if } \mathbf{b}\cdot \mathbf{z}_i = s_i \quad \forall i \in \{1,\dots,r\} \\ 0 & \text{otherwise}. \end{cases}
\end{equation}
Substituting this overlap confirms our general formula. In the trivial case where $\mathbf{z}=\mathbf{0}$ (the Identity operator), we find $a_{P(\mathbf{0},\mathbf{0})} = 2^{r-n} \cdot 2^{n-r} = 1$, confirming trace preservation.

\subsection{Limiting Case: Single-Qubit Rotations and Multi-Control Z Gates}
We further benchmark our formula against previously known bounds for specific diagonal gates. Without loss of generality, let $U_k = \text{diag}(1, e^{2\pi i/2^k})$ be a phase shift gate acting on the $n^{\text{th}}$ qubit. The generalized Pauli strings can be factored as $P = Q \otimes P_{n^{\text{th}}}$, where $Q \in \mathcal{P}_{n-1}$. Under similarity transformation by $U_k$, the Pauli expansion splits strictly into two families based on commutation~\cite{Garcia_2023}:

\begin{enumerate}
    \item \textbf{Commuting Set} ($\mathcal{C}$): For $P_{n^{\text{th}}} \in \{I, Z\}$, $[P, U_k] = 0$. Since $|\psi\rangle$ is a stabilizer state, $|a_{P}|^2 \in \{0, 1\}$.
    \item \textbf{Anticommuting Set} ($\mathcal{A}$): For $P_{n^{\text{th}}} \in \{X, Y\}$, the similarity transformation generates a continuous rotation in the $X$-$Y$ plane:
\begin{equation}
    U_k^\dagger X_n U_k = \cos\!\left(\frac{2\pi}{2^k}\right) X_n - \sin\!\left(\frac{2\pi}{2^k}\right) Y_n.
\end{equation}
This yields expectation values strictly proportional to $\cos^2(2\pi /2^k)$ or $\sin^2(2\pi /2^k)$ depending on whether $Q \otimes X_n$ or $Q \otimes Y_n$ resides in the stabilizer group.
\end{enumerate}

Our analytical polynomial identically recovers this behavior. For a single-qubit phase shift, the polynomial term simplifies to $e^{i(\theta(\mathbf{b})-\theta(\mathbf{b}\oplus \mathbf{x}))} = \exp\!\left(2\pi i \frac{b_n - (b_n \oplus x_n)}{2^k}\right)$. If $x_n = 0$ (the commuting set), the phase evaluates to $1$, yielding purely boolean amplitudes $|a_{P(\mathbf{x},\mathbf{z})}|^2 \in \{0,1\}$. If $x_n = 1$ (the anticommuting set), evaluating the magnitude squared gives:
\begin{equation}
|a_{P(\mathbf{x},\mathbf{z})}|^2 = \begin{cases}
    \cos^2(2\pi/2^k) & \text{for } z_n = 0,\, (\mathbf{x},\mathbf{z}) \in \mathbb{L}\\
    \sin^2(2\pi/2^k) & \text{for } z_n = 1,\, (\mathbf{x},\mathbf{z}) \notin \mathbb{L}\\
    0 & \text{otherwise},
\end{cases}
\end{equation}
in perfect agreement with the geometric derivation.

Finally, for the higher-hierarchy multi-controlled gate $C^{k-1}Z$ acting on the initial graph state $|+\rangle^{\otimes k}$, the phase polynomial takes the canonical Boolean form $e^{i(\theta(\mathbf{b})-\theta(\mathbf{b}\oplus \mathbf{x}))} = \exp\!\left( \pi i \left[\prod_{i=1}^k b_i - \prod_{j=1}^k (b_j+x_j)\right] \right)$. Because $|+\rangle^{\otimes k}$ has only pure-$X$ stabilizers ($r=0$), our analytical formula smoothly collapses to:
\begin{equation}
    a_{P(\mathbf{x},\mathbf{z})} = \frac{1}{2^k} \sum_{\mathbf{b} \in \mathbb{Z}_2^k} (-1)^{\mathbf{z}\cdot \mathbf{b}} \exp\!\left( \pi i \left[\prod_{i=1}^k b_i - \prod_{j=1}^k (b_j+x_j)\right] \right),
\end{equation}
which identically reproduces the Walsh-Hadamard transform representations foundational to Ref.~\cite{Beverland_2020}.

\section{Magic Bounds of The Shallow Layer Model}
\label{app:shallow_bound}

Since this geometric bounding relies entirely on the underlying phase polynomial structure of the diagonal hierarchy gate $U \in \mathcal{D}_n^{(k)}$, these bounds are strictly universal and completely independent of the specific polynomial coefficients used to define a given magic measure, unlike the case of the Clifford Hierarchy.

\subsection{Lower Bound: Zero Magic Remains Possible}
 Zero magic generation here implies that $U$ maps a stabilizer to another stabilizer state, which is the upper bound of $F_\alpha$. 

\begin{theorem}[Stabilizer  Maximum bounds and Zero-Magic]
\label{thm:zero_magic_diagonal}
Let $|\psi\rangle \in \mathrm{STAB}_n$ be an $n$-qubit stabilizer state with $r$ independent pure-$Z$ generators.

\begin{enumerate}
    \item \textbf{State-to-gate direction.}
    If $r \ge k-2$, then there exists a diagonal hierarchy gate 
$    U \in \mathcal D_n^{(k)} \setminus \mathcal D_n^{(k-1)}$
such that
$  U|\psi\rangle \in \mathrm{STAB}_n.$

    \item \textbf{Gate-to-state direction.}
    Conversely, let $U \in \mathcal D_n^{(k)} \setminus \mathcal D_n^{(k-1)}.$
    If $n \ge k-1$, then there exists a nontrivial stabilizer state
    $
    |\psi\rangle \in \mathrm{STAB}_n,
    $ (not computational basis)
    such that
    $U|\psi\rangle \in \mathrm{STAB}_n.$
\end{enumerate}
\end{theorem}

\begin{proof}
We first observe that the number of pure-$Z$ generators, $r$, restricts the dimension of the state's coherent superposition. Any $n$-qubit stabilizer state with $r$ independent pure-$Z$ generators lies in the orbit of the canonical state
$|\phi\rangle = |0\rangle^{\otimes r}\otimes |+\rangle^{\otimes (n-r)}$
under the subgroup of Clifford operations that preserves the pure-$Z$ sector, namely those with symplectic form
\begin{equation}
    M=\begin{pmatrix}A&0\\ C&D\end{pmatrix},
\end{equation}
generated by $\mathrm{CX}$, $\mathrm{CZ}$, and $S$ gates. Let $U_C$ be such a Clifford operation such that
$|\psi\rangle = U_C |\phi\rangle.$
Since this subgroup normalizes the diagonal hierarchy, we obtain
\begin{equation}
    U|\psi\rangle \in \mathrm{STAB}_n
\iff
V|\phi\rangle \in \mathrm{STAB}_n,
\end{equation}
where
$V = U_C^\dagger U U_C \in \mathcal{D}_n^{(k)}\setminus \mathcal{D}_n^{(k-1)}$ remains a strictly level-$k$ diagonal hierarchy gate.

The action of $V$ on $|\phi\rangle$ applies a phase polynomial to the computational basis states. Because the first $r$ qubits are fixed in the $|0\rangle$ state, the phase function is restricted to the remaining $n-r$ variables. The output state takes the form
\begin{equation}
    V|\phi\rangle = |0\rangle^{\otimes r} \otimes  \frac{1}{\sqrt{2^{n-r}}} \sum_{b_x \in \mathbb{Z}_2^{n-r}} \exp\Biggl(2\pi i \sum_m \frac{f_m(\mathbf 0^r,\mathbf b_x)}{2^m}\Biggr) |\mathbf b_x\rangle ,
\end{equation}
where $f_m(0^r,b_x)$ denotes the restricted Boolean polynomial.

To make the mechanism transparent, we isolate the highest hierarchy contribution and write
\begin{equation}
    f(\mathbf b)=\sum_{m=1}^{k}\frac{f_m(\mathbf b)}{2^m},
\qquad
f_m:\mathbb Z_2^n\to \mathbb Z_{2^m},
\end{equation}
with $f_k\not\equiv 0$, so that the gate is genuinely level $k$. Splitting the variables as $\mathbf b=(\mathbf y,\mathbf b_x)$, where $\mathbf y\in\mathbb Z_2^r$ are fixed to $0^r$ on the support of the input stabilizer state, the restricted phase is $f(\mathbf0^r,\mathbf b_x)$.

For the output to remain a stabilizer state, this restricted phase must reduce modulo $1$ to a Clifford phase, equivalently to a $\mathbb Z_4$-valued quadratic form. Thus, the role of the frozen variables is to eliminate the excess non-Clifford structure carried by the level-$k$ term. In particular, if the highest-level contribution is chosen to depend on at least $k-2$ frozen variables, then setting $\mathbf y = 0^r$ annihilates that term, while the restricted phase remains Clifford-compatible. This yields the sufficient condition $r\ge k-2$.

The same mechanism can also be read in the opposite direction. For any fixed genuinely level-$k$ diagonal gate
\begin{equation}
    U \in \mathcal D_n^{(k)} \setminus \mathcal D_n^{(k-1)},
\end{equation}
if $n \ge k-1$, then one can always choose a nontrivial stabilizer input, namely one that is not a computational-basis state and therefore satisfies $r\le n-1$, such that
\begin{equation}
    U|\psi\rangle \in \mathrm{STAB}_n.
\end{equation}
Combining the condition from Theorem~\ref{thm:zero_magic_shallow} with the nontriviality condition $r\le n-1$ shows that such a stabilizer input exists whenever $n\ge k-1$. Thus, the zero-magic phenomenon is not only a property of certain input families at a fixed hierarchy level, but also a generic obstruction for any fixed diagonal gate once the input stabilizer geometry is chosen appropriately.
\end{proof}

The minimum magic generation for an arbitrary stabilizer state in our shallow-circuit model implies that there always exists a diagonal hierarchy gate $U \in \mathcal{D}_n^{(k)} \setminus \mathcal{D}_n^{(k-1)}$ return input stabilizer state back the stabilizer polytope with $r \geq k-2$. This establishes the absolute lower bound for our operational magic family, where the underlying phase polynomial functional achieves its theoretical maximum, $F_{\alpha}(U|\psi\rangle) = 2^{n\alpha}$, yielding a minimum macroscopic magic value of $\mathcal{Q}(U|\psi\rangle) = 0$ for all valid operational measures.

\subsection{Upper Bound: Maximal Magic Requires Graph-State Geometry}
We now invert this perspective to investigate the maximum possible magic generation within this shallow-layer model. For any given experimental measure within our operational family, this corresponds to identifying the absolute upper bound of the quantum magic $\mathcal{Q}(U|\psi\rangle)$, which is mathematically driven by evaluating the absolute lower bound of the underlying functional $F_\alpha$ for all $\alpha \geq 2$.

\begin{theorem}[Flat Minimum and the Graph-State Condition]
\label{thm:flat_minimum_graph_state}
Let $|\psi\rangle \in \mathrm{STAB}_n$ be an $n$-qubit stabilizer state characterized by $r$ linearly independent pure $Z$-generators, and let $U \in \mathcal{D}_n^{(k)} \setminus \mathcal{D}_n^{(k-1)}$ be a strictly level-$k$ diagonal hierarchy gate. Consider the operational magic functional in Eq.~\eqref{eq:magic polynominal}
subject to the normalized constraint 
\begin{equation}
    \sum_{\mathbf{x},\mathbf{z}} |a_{P(\mathbf{x},\mathbf{z})}|^2 = 2^n
\end{equation}

If the functional $F_\alpha$ attains its global minimum under the aforementioned quadratic constraint, then the following structural conditions must hold:
\begin{enumerate}
    \item \textbf{Spectrum Flatness:} The Pauli spectrum is perfectly flat across all valid non-identity elements, satisfying
    \begin{equation}
            |a_{P(\mathbf{x},\mathbf{z})}|^2 = \frac{2^n - 1}{4^n - 2^n} \qquad \text{for all } (\mathbf{x} \neq \mathbf{0}^n, \mathbf{z}) \in \mathbb{Z}_2^{2n};
    \end{equation}

    \item \textbf{Generator Constraint:} The initial stabilizer state must strictly satisfy $r=0$, meaning the stabilizer group $\mathcal{S}$ contains no non-trivial pure $Z$-elements;
    \item \textbf{Graph-State Equivalence:} Under a change of stabilizer generators, the initial state admits generators of the canonical form
    \begin{equation}
                g_i = (-1)^{s_i} X_i \prod_{j=1}^n Z_j^{A_{ij}}, \qquad i=1,\dots,n,
    \end{equation}
    where the adjacency matrix satisfies $A = A^T \in \mathcal{M}_n(\mathbb{Z}_2)$. Consequently, the initial state is strictly of graph-state type, and in particular, is local-Clifford equivalent to a graph state.
\end{enumerate}
\end{theorem}

\begin{proof}
To minimize the strictly convex operational functional $F_\alpha$, Jensen's inequality dictates that the Pauli spectrum must be distributed as uniformly as possible. As established in Eq.~\eqref{eq: shallow_model_pauli_spectrum}, the Kronecker delta imposes $r$ linearly independent geometric constraints on the summation, artificially confining the spectrum's support to an $(n-r)$-dimensional subspace. To eliminate this bottleneck and maximally disperse the probability distribution, we must strictly set $r=0$. 

To explicitly evaluate this flat minimum, we note that diagonal operators commute with pure $Z$-type Pauli strings. Consequently, the $\mathbf{x} = \mathbf{0}^n$ subspace remains invariant: $|a_{P(\mathbf{0}^n,\mathbf{z})}|^2 = |\langle \psi | Z(\mathbf{z}) | \psi \rangle|^2$. Because an $r=0$ state contains no non-trivial pure $Z$-elements, this yields exactly $1$ for the identity and $0$ otherwise. Applying the purity constraint $\sum |a_{P}|^2 = 2^n$, the remaining probability weight for the non-identity subspace is exactly $2^n - 1$. Distributing this weight uniformly across the $4^n - 2^n$ available non-zero elements enforces the spectrum flatness condition:
\begin{equation}
       |a_{P(\mathbf{x} \neq \mathbf{0}^n,\mathbf{z})}|^2 = \frac{2^n - 1}{4^n-2^n} = 2^{-n}.
    \label{eq: flatness}
\end{equation}

Any $n$-qubit stabilizer state satisfying $r=0$ can be constructed in the symplectic formalism by applying diagonal Clifford operations to the unentangled product state $|+\rangle^{\otimes n}$. Because diagonal Clifford gates act only to shift degree-2 phase polynomials, they natively generate graph states when applied to $|+\rangle^{\otimes n}$. Thus, the optimal initial configuration for maximum magic generation is strictly local-Clifford equivalent to a graph state.
\end{proof}

Substituting this perfectly flat distribution into the operational functional establishes the absolute lower bound:
\begin{equation}
    F_\alpha(U|\psi\rangle) \geq 1 + (2^{n} - 1) 2^{n (1-\alpha)}.
\end{equation}This theoretical minimum is physically saturable for specific even-parity or single-qubit systems. For $n=1$, applying the $T$ gate to $|+\rangle$ exactly yields $F_\alpha = 1 + 2^{1-\alpha}$. For $n=2$, applying the controlled-phase ($CS$) gate to $|+\rangle^{\otimes 2}$ generates a perfectly flat non-identity spectrum, saturating the bound at $F_\alpha = 1 + 12 \cdot 2^{-2\alpha}$.

\section{\texorpdfstring{$N$}{N}-layer SQR Model}
\label{app:deep_STAR_analytic}
Motivated by the STAR architecture, a key design principle is to avoid the substantial space–time overhead associated with magic-state distillation.~\cite{PRXQuantum.5.010337} Instead, the non-Clifford layer is implemented using arbitrary-angle Pauli-$Z$ rotations as native operations, enabled by the $[[4,1,1,2]]$ subsystem code together with relatively clean state injection. By considering one common choice (SQR), diagonal gates in the Clifford hierarchy admit a simplified representation of the form
\begin{equation}
    U = \sum_{\mathbf{b}\in \mathbb{Z}_2^n} \exp\left(2\pi i (\mathbf{w} \cdot \mathbf{b})\right) |\mathbf{b}\rangle\langle\mathbf{b}| 
    \label{eq:STAR_phase_poly}
\end{equation}
where the phase weights satisfy $2^k\cdot w\in\{0,1,\dots,2^k-1\}^n$, each component is a rational number with denominator $2^k$. 

In this setting, the single-qubit nature of the native rotations constrains the phase polynomial in Eq.~\eqref{eq:STAR_phase_poly} to be linear in the Boolean variables. As a consequence, the shallow-layer model (ansatz~1) simplifies significantly, and the corresponding Pauli spectrum takes the form
\begin{align}
a_{P(\mathbf x,\mathbf z)} =  \frac{2^r}{2^n} C(x) \sum_{\substack{\mathbf b \in \mathbb{Z}_2^n \\ \mathbf b Z = \mathbf h' \oplus \mathbf K}} \exp\left( 2\pi i \, V(\mathbf x, \mathbf z) \cdot \mathbf \mathbf b \right) 
\label{eq: shallow_STAR_spectrum}
\end{align}
where the prefactor is given by $C(x) = \exp(2\pi i( w \cdot x)) (-1)^{s_0+x\cdot z_{ref}}i^{x\cdot z_{ref}}$ and and we have introduced an effective phase vector
\begin{equation}
    \mathbf{V}(\mathbf{x}, \mathbf{z}) = 2(\mathbf{w} \circ \mathbf{x}) + \frac{\mathbf{z} \oplus \mathbf{z}_{\mathrm{ref}}}{2}
\end{equation}
defined via the Hadamard (element-wise) product $\circ$.

We now proceed to the second layer, consisting of a Clifford operation followed by a diagonal (non-Clifford) layer. A Clifford unitary acts on Pauli operators by symplectic transformation, mapping each Pauli string to another Pauli string up to a phase factor $(-1)^{h(x,z)}$~\cite{gottesman_knill,Aaronson_2004}. Accordingly, if the density matrix before the Clifford layer is expanded in the Pauli basis, the transformed state can be written as 
\begin{align}
        \rho' &= \sum_{(\mathbf x,\mathbf z)\in \mathbb{F}_2^{2n}} a_{P(\mathbf x,\mathbf z)} (-1)^{h(\mathbf x,\mathbf z)}P(F(\mathbf x,\mathbf z)) 
        \notag \\
        &= \sum_{(\mathbf x',\mathbf z')\in \mathbb{F}_2^{2n}} a_{P(F^{-1}(\mathbf x',\mathbf z'))} (-1)^{h(F^{-1}(\mathbf x',\mathbf z'))}P(\mathbf x',\mathbf z')
\end{align}
where $F \in \mathrm{Sp}(2n)$ denotes the induced symplectic transformation on binary phase-space variables~\cite{Dehaene_2003}.

We next apply an arbitrary-angle Pauli-$Z$ rotation layer $U$. Such diagonal unitaries act nontrivially only on the $X$ components of Pauli strings, yielding $UP(\mathbf x',\mathbf z')U^\dagger = VP(\mathbf x',\mathbf z')$ where $V$ encodes the induced phase dressing~\cite{Cui_2017}. Exploiting the linear Boolean structure of the phase polynomial~\cite{Amy_2013}, one can rewrite the computational-basis dependence in terms of Pauli-$Z$ operators via the correspondence $\mathbf b \mapsto (I - Z)/2$~\cite{Gross_2006}. This leads to the expansion
\begin{align}
    UP(\mathbf x',\mathbf z')U^\dagger &=  \exp(2\pi i \sum_{j=1}^n w_j x'_j Z_j)P(\mathbf x',\mathbf z') = \sum_{\mathbf u \leq \mathbf x'} P(\mathbf x',\mathbf z') Z^{\mathbf u} \prod_{j: x'_j=1} \left( \cos(2\pi w_j)^{1-u_j} (i \sin(2\pi w_j))^{u_j} \right) \notag \\
    &= \sum_{\mathbf u \leq \mathbf x'} P(\mathbf x',\mathbf z'\oplus \mathbf u) i^{\mathbf x\cdot \mathbf z - \mathbf x\cdot (\mathbf z\oplus \mathbf u) + \mathbf x\cdot \mathbf u} \prod_{j: x'_j=1} \left( \cos(2\pi w_j)^{1-u_j} (i \sin(2\pi w_j))^{u_j} \right) \notag \\
    &=\sum_{\mathbf u\leq \mathbf x'} \Gamma_{\mathbf x'}(\mathbf u) (-1)^{\mathbf u\cdot \mathbf z'+|\mathbf u|} P(\mathbf x',\mathbf z'\oplus \mathbf u)
\end{align}
where the sum runs over all binary vectors $u \in \mathbb{Z}_2^n$ supported on the support of $x'$ (i.e., $u_j = 0$ whenever $x'_j = 0$), and the coefficients are given by
\begin{equation}
     \Gamma_{\mathbf x'}(\mathbf u) = \prod_{j: x'_j=1} \left( \cos(2\pi w_j)^{1-u_j} ( \sin(2\pi w_j))^{u_j} \right)
\end{equation}

Combining the Clifford and diagonal layers, the Pauli expectation values after the second layer of our model (ansatz~2) take the form by replacing variable with $x = x'$ and $z = z'\oplus u$
\begin{equation}
    a'_{P(\mathbf x,\mathbf z)} = \sum_{\mathbf u \le \mathbf x} \Gamma_{x}(\mathbf u) (-1)^{|\mathbf u|+\mathbf u\cdot \mathbf z+h(F^{-1}(\mathbf x, \mathbf z \oplus \mathbf u))} a_{P(F^{-1}(\mathbf x, \mathbf z \oplus \mathbf u))}
\end{equation}
By iterating this transformation over $N$ alternating Clifford and analog diagonal-rotation layers, one obtains the full evolution of the logical state within the $N$-layer model, motivated by early FTQC and STAR architectures.

\section{No-go Theorem: State-Independent Gate-selection Rule}
\label{app:no_go_2}
To show that there exists a state-independent gate choice that can guarantee magic generation, we collect the Pauli expectation values into a real vector
\begin{equation}
    \mathbf{a}
    =
    \{a_{P(\mathbf x,\mathbf z)}\}
    \in \mathbb{R}^{4^n},
\end{equation}
then the update induced by one Clifford-plus-rotation block can be written abstractly as $\mathbf{a}\longmapsto \mathbf{a}'=M\mathbf{a}$.

Because the Pauli expectation values satisfy the normalization
\begin{equation}
    \|\mathbf{a}\|_2^2 = \sum_{x,z}|a_{P(x,z)}|^2 = 2^n,
\end{equation}
the transfer matrix $M$ preserves the Euclidean norm and therefore acts orthogonally on the Pauli vector. Now, we assume such a linear mapping exists and is independent of state choice.

\begin{theorem}[No-go for universal Pauli-spectrum flattening under orthogonal transfer]
\label{thm:no_go_universal_flattening_supplementary}
Let $d=4^n$, and let $\mathbf{a}\in\mathbb{R}^d$ denote the Pauli-coefficient vector of a pure $n$-qubit state, normalized so that
\begin{equation}
    \|\mathbf{a}\|_2^2 = 2^n.
\end{equation}
Let $M\in O(d)$ be a real orthogonal matrix acting linearly on $\mathbf{a}$,
\begin{equation}
    \mathbf{a}' = M\mathbf{a}.
\end{equation}
For any $\alpha>1$, define
\begin{equation}
    F_\alpha(\mathbf{a}) := \sum_{j=1}^d |a_j|^{2\alpha}
    = \|\mathbf{a}\|_{2\alpha}^{2\alpha}.
\end{equation}
Then there does not exist an orthogonal matrix $M$ such that
\begin{equation}
    F_\alpha(M\mathbf{a}) \le F_\alpha(\mathbf{a})
\end{equation}
for all pure-state Pauli vectors $\mathbf{a}$ with $\|\mathbf{a}\|_2^2=2^n$, with strict inequality on a set of nonzero measure. Equivalently, an orthogonal transformation cannot universally flatten the Pauli spectrum: if it decreases $F_\alpha$ for some states, it must increase $F_\alpha$ for others.
\end{theorem}

\begin{proof}
Let
\begin{equation}
    S := \left\{\mathbf{a}\in\mathbb{R}^d : \|\mathbf{a}\|_2^2 = 2^n \right\}
\end{equation}
be the $L_2$ sphere of radius $2^{n/2}$ in $\mathbb{R}^d$. Since $M\in O(d)$, it preserves the Euclidean norm,
\begin{equation}
    \|M\mathbf{a}\|_2 = \|\mathbf{a}\|_2,
\end{equation}
and therefore maps $S$ bijectively onto itself. Moreover, as an orthogonal transformation, $M$ preserves the rotationally invariant surface measure $d\mu$ on $S$~\cite{Bengtsson_Zyczkowski_2006}.

Assume, toward a contradiction, that
\begin{equation}
    F_\alpha(M\mathbf{a}) \le F_\alpha(\mathbf{a})
    \qquad \text{for all } \mathbf{a}\in S,
\end{equation}
and that the inequality is strict on a measurable subset $E\subset S$ with $\mu(E)>0$. Since $F_\alpha$ is continuous and $M$ is continuous, the function
\begin{equation}
    g(\mathbf{a}) := F_\alpha(\mathbf{a}) - F_\alpha(M\mathbf{a})
\end{equation}
is continuous on $S$, nonnegative everywhere, and strictly positive on $E$. Hence
\begin{equation}
    \int_S g(\mathbf{a})\, d\mu(\mathbf{a}) > 0,
\end{equation}
that is,
\begin{equation}
    \int_S F_\alpha(\mathbf{a})\, d\mu(\mathbf{a})
    >
    \int_S F_\alpha(M\mathbf{a})\, d\mu(\mathbf{a}).
\label{eq:strict_contradiction}
\end{equation}

On the other hand, since $M: S \to S$ is a measure-preserving bijection (owing to an orthogonal matrix), the change of variables $\mathbf{b}=M\mathbf{a}$ gives
\begin{equation}
    \int_S F_\alpha(M\mathbf{a})\, d\mu(\mathbf{a})
    =
    \int_S F_\alpha(\mathbf{b})\, d\mu(\mathbf{b})
    =
    \int_S F_\alpha(\mathbf{a})\, d\mu(\mathbf{a}),
\label{eq:measure_preserving}
\end{equation}
which contradicts Eq.~\eqref{eq:strict_contradiction}.

Therefore, no orthogonal transformation can satisfy
\begin{equation}
    F_\alpha(M\mathbf{a}) \le F_\alpha(\mathbf{a})
\end{equation}
for all $\mathbf{a}\in S$ with strict inequality on a set of nonzero measure. In other words, a universal flattening of the Pauli spectrum is impossible under any state-independent orthogonal transfer. Whenever an orthogonal transformation decreases $F_\alpha$ for some input states, it must necessarily increase $F_\alpha$ for others.
\end{proof}

These results not only show that, in the SQR gate setup, no such state-independent selection rule exists. It actually generalized to any mapping between two different Pauli spectra. Once you follow the normalization condition, there is no way to find a forbidden or allowed gate choice. The only choice you can find is equal, which is the hypercube group, in our setting, corresponding to the Clifford group. In the $N$-layer model (ansatz~2), you will see that it corresponds to a $\mathbb{Z}_4$ rotation.

\section{Shallow SQR Pauli Support Size}
We now turn to the central question: to what extent can the SQR setup generate maximal quantum magic? As established earlier, achieving a perfectly flat Pauli distribution requires satisfying two distinct and independent mathematical conditions:
\begin{enumerate}
\item \textit{Support maximality}: the unitary evolution must possess sufficient geometric expressibility to distribute weight over all $4^n - 2^n$ non-identity Pauli operators.

\item \textit{Amplitude compatibility}: the resulting Pauli coefficients must obey the rigid integer constraints imposed by probability normalization.
\end{enumerate}
These two requirements correspond, respectively, to the \emph{Pauli rank} (support size) and the \emph{flatness} (uniformity of amplitudes) of the Pauli spectrum. In this section, we first analyze the former by determining the maximal achievable support under the $N$-layer model(ansatz~2).

To this end, we define the Pauli support size
\begin{equation}
\mathcal N = \sum_{\mathbf{x} \in \mathbb{Z}_2^n}
\bigl|{\mathbf{z} \in \mathbb{Z}_2^n \mid a_{P(\mathbf{x},\mathbf{z})} \neq 0}\bigr|,
\end{equation}
which counts the number of nonzero Pauli expectation values and thus quantifies the Pauli rank.

To probe the maximal spreading capability, we consider the input state $|+\rangle^{\otimes n}$, which is known to maximize magic generation within the stabilizer class (Theorem~\ref{thm:flat_minimum_graph_state}). For ansatz~2, the resulting Pauli spectrum admits a factorized form
\begin{align} 
a_{P(x,z)} &= \frac{1}{2^n} \sum_{\mathbf{b} \in \mathbb{Z}_2^n} \exp\left( 2\pi i \sum_{j=1}^n \left(2w_j x_j + \frac{z_j}{2}\right) b_j \right) \notag \\ &=  \prod_{j=1}^n \frac{1 + \exp\left( 2\pi i \left( 2w_j x_j + \frac{z_j}{2} \right) \right)}{2} 
\end{align}
This product structure reflects the strictly local nature of the diagonal layer in the ansatz~2 and will play a crucial role in constraining the achievable Pauli rank.

To determine when $a_{P(x,z)} \neq 0$, we note that every factor in the product representation must be nonvanishing. A given factor vanishes if and only if its phase satisfies
\begin{equation}
2w_j x_j + \frac{z_j}{2} \equiv \frac{1}{2} \pmod{1}.
\end{equation}
We therefore count, for each site $j$, the number of valid choices of $z_j \in {0,1}$ that avoid this condition, under the constraint $w_j = k_j / 2^m$.

\begin{enumerate}
\item If $x_j = 0$, the condition reduces to $\frac{z_j}{2} \equiv \frac{1}{2} \pmod{1}$, which is satisfied only by $z_j = 1$. Hence, exactly one valid choice remains, namely $z_j = 0$.

\item If $x_j = 1$, the condition becomes
\begin{equation}
    2w_j + \frac{z_j}{2} \equiv \frac{1}{2} \pmod{1}.
\end{equation}
The number of admissible values of $z_j$ depends on $w_j$: If $4w_j \in \mathbb{Z}$, exactly one choice of $z_j$ satisfies the condition, leaving one valid option; If $4w_j \notin \mathbb{Z}$, neither $z_j = 0$ nor $z_j = 1$ satisfies the condition, and both choices remain valid.
\end{enumerate}

Let $S_{\mathbf{w}}$ denote the set of indices $j$ such that $4w_j \notin \mathbb{Z}$. Equivalently, for $w_j = k_j/2^m$, this corresponds to $k_j$ not being divisible by $2^{m-2}$. For a given binary vector $\mathbf{x}$, define
\begin{equation}
k(\mathbf{x}) = \sum_{j \in S_{\mathbf{w}}} x_j,
\end{equation}
which counts the number of active components contributing two valid choices for $z_j$. The Pauli support size can then be expressed as
\begin{equation}
\mathcal N(SQR) = \sum_{\mathbf{x} \in \mathbb{Z}_2^n} 2^{k(\mathbf{x})}.
\end{equation}

Exploiting the independence across sites, this sum factorizes into a closed-form expression. Let $K = |S_{\mathbf{w}}|$. For each $j \in S_{\mathbf{w}}$, summing over $x_j \in {0,1}$ yields $1 + 2 = 3$, while for $j \notin S_{\mathbf{w}}$, it yields $1 + 1 = 2$. Consequently, the Pauli rank is given exactly by
\begin{equation}
\mathcal N = 3^K 2^{n-K}.
\end{equation}

This result establishes a fundamental limitation of the SQR gate choice in ansatz~2. Even in the optimal case $K = n$, the maximal support is bounded by
\begin{equation}
\mathcal N_{\max}(SQR) = 3^n < 4^n - 2^n \quad (n > 1),
\end{equation}
demonstrating that a perfectly flat Pauli distribution is unattainable beyond the single-qubit case. Indeed, only for $n=1$ (e.g., the $T$ gate) can the theoretical bound be saturated.

By contrast, architectures that allow higher-order (nonlinear) phase functions can, in principle, achieve the maximal support $\mathcal N_{\max} = 4^n - 2^n$. This highlights a key architectural implication: restricting diagonal operations to single-qubit rotations fundamentally limits Pauli spreading and, consequently, the achievable quantum magic. Incorporating multi-qubit diagonal interactions is therefore essential for reaching the maximal expressibility required for optimal magic generation.

\end{document}